\documentclass[11pt,3p]{scrartcl}

\usepackage[T1]{fontenc}
\usepackage[utf8]{inputenc}
\usepackage{lmodern}
\usepackage[english]{babel}
\usepackage{xcolor}
\usepackage{amsmath}
\usepackage{amssymb}
\usepackage{amsfonts}
\usepackage{amsthm}
\usepackage[centerdot]{mathtools}
\usepackage{hyperref}
\usepackage{tikz}
\usepackage[ruled]{algorithm2e}
\SetKw{Continue}{continue}
\SetKw{Null}{null}
\SetKw{Break}{break}
\usepackage[textsize=tiny]{todonotes}
\usepackage[normalem]{ulem}
\usepackage{float}
\usepackage{subcaption}
\usepackage{authblk}

\usepackage[nameinlink]{cleveref}

\newtheoremstyle{thm}% name
{15pt}% Space above
{5pt}% Space below 
{\itshape}% Body font
{}% Indent amount: Indent amount: empty = no indent, \parindent = normal paragraph indent
{\bfseries}% Theorem head font\part{title}
{}% Punctuation after theorem head
{0.5em}% Space after theorem head: Space after theorem head: { } = normal interword space; \newline = linebreak
{}% Theorem head spec (can be left empty, meaning `normal')
\theoremstyle{thm}% default

\usepackage{xpatch}

\xpatchcmd{\proof}{\itshape}{\prooflabelfont}{}{}
\newcommand{\prooflabelfont}{\bfseries}

\newtheorem{thm}{Theorem}[section]
\newtheorem{lem}[thm]{Lemma}
\newtheorem{definition}[thm]{Definition}
\newtheorem{property}[thm]{Property}

\title{Finding all minimum cost flows and a faster algorithm for the $K$ best flow problem}

\author[1]{David~K\"onen}
\author[2]{Daniel~Schmidt}
\author[ ]{Christiane~Spisla\footnote{Electronic addresses: \href{mailto:david.koenen@uni-koeln.de}{david.koenen@uni-koeln.de},  \href{mailto:daniel.schmidt@uni-bonn.de}{daniel.schmidt@uni-bonn.de}, \href{mailto:c.spisla@gmx.de}{c.spisla@gmx.de}  }}
\affil[1]{%
	Department of Supply Chain Management \& Management Science\\
	University of Cologne\\
	Albertus-Magnus-Platz\\
	50923 Cologne, Germany
}

\affil[2]{%
	Institut für Informatik V\\
	Universität Bonn\\
	Friedrich-Hirzebruch-Allee 5\\
	53115 Bonn, Germany
}

\begin{document}

\maketitle

\begin{abstract}
  This paper addresses the problem of determining all optimal integer solutions of a linear integer network flow problem, which we call the \emph{all optimal integer flow (AOF)} problem.  We derive an $\mathcal{O}(F(m+n)+mn+M)$ time algorithm to determine all $F$ many optimal integer flows in a directed network with $n$ nodes and $m$ arcs, where $M$ is the best time needed to find one minimum cost flow. We remark that stopping Hamacher's well-known method for the determination of the $K$ best integer flows \cite{Hamacher1995} at the first sub-optimal flow results in an algorithm with a running time of $\mathcal{O}(Fm(n\log n +m)+ M)$ for solving the AOF problem. Our improvement is essentially made possible by replacing the shortest path sub-problem with a more efficient way to determine a so-called \emph{proper zero cost cycle}  using a modified depth-first search technique.   
  As a byproduct, our analysis  yields an enhanced algorithm to determine the $K$ best integer flows that runs in $\mathcal{O} (Kn^3+M)$.  Besides, we give lower and upper bounds for the number of all optimal integer and feasible integer solutions. Our bounds are based on the fact that any optimal solution can be obtained by an initial optimal \emph{tree solution} plus a conical combination of incidence vectors of all \emph{induced cycles} with bounded coefficients.  
\end{abstract}

	\par\vskip\baselineskip\noindent

\newpage 

\section{Introduction}
\label{introduction}

Given a capacitated directed graph with $n$ nodes $V$ and $m$ edges $E$, edge costs, and a flow balance value $b_v$ for each node $v \in V$, a $b$-flow is a network flow that respects the capacities, and each node $v \in V$ has flow balance exactly $b_v$. 
The minimum cost flow (MCF) problem asks for a $b$-flow of minimum total cost.
The problem is one of the fundamental problems in combinatorial optimization.
Among others, the important \emph{minimum-weight perfect matching problem in bipartite graphs} and the \emph{transportation problem} reduce to the minimum cost flow problem. Consequently, the problem models many basic building blocks in a wide range of applications in industry and decision making.
As proof, Chapter 19 of~\cite{Ahuja1993} details applications from karyotyping of chromosomes, building evacuations,  warehouse layouting, inventory planning, and project management.
Solving minimum cost flow problems is a common task in many branch-and-cut algorithms.

As such, the problem is well-studied, and a wide range of polynomial algorithms is known. 
We refer to~\cite{Ahuja1993} for an extensive overview and merely mention that on an input graph with $n$ nodes and $m$ edges, the \emph{enhanced capacity scaling algorithm}~\cite{Orlin1993} solves the problem in time $O\bigl((m \log n)(m + n\log n)\bigr)$, i.e., in strongly polynomial time. 
The network simplex algorithm is a specialization of Dantzig's simplex algorithm for linear programming and has proven efficient in practice.

However, being a fundamental and abstract model, minimum cost flows tend to ignore certain aspects of applications. 
For instance, drawing a planar graph on a rectangular grid with a minimum number of  edge bends can be modeled as a minimum cost flow problem~\cite{Battista1998,Neta2012,Juenger2004}.
The result of this bend-minimization process affects indirectly and not predictably other properties of the final drawing, like the drawing area, the aspect ratio or the total edge length. 
So different minimum cost flows can yield graph drawings all with the minimum number of bends but completely different appearances.
The above criteria (and others) define the readability and \emph{goodness} of a drawing, but eventually, whether a given drawing is good lies in the eye of the beholder.
In situations like these, devising algorithms is a non-trivial task; particularly, if no formal model for adherence exists and recognizing suitable solutions is up to the user.
Here, a quick hands-on answer is to enumerate \emph{all} minimum cost flows, i.e., all possibilities for a bend minimal drawing, and let the user decide on the details. 

Christofides~\cite{Christofides1986} proposes an algorithm that enumerates all \emph{basic} integer minimum cost flows. 
However, the algorithm may enumerate the same solution multiple times.
We are interested in a more efficient enumeration procedure.
Here, the best we can hope for is an algorithm whose worst-case running time is polynomial in $n,m$ \emph{and} the number of enumerated flows, which generally, might be exponential in $n$ and $m$.
To ensure that each solution is only enumerated once, Hamacher~\cite{Hamacher1995} proposes a binary partition approach and considers the $K$ best integer $b$-flow problem: If $f_1,f_2,\dots$ are all feasible $b$-flows in some order of non-decreasing cost, Hamacher's algorithm will return $f_1,\dots, f_K$ in time $O\bigl(Km (n \log n + m) + M\bigr)$, where $M$ is the time needed to find \emph{one} minimum cost flow.
Sedeño-Noda and Espino-Mart\'in~\cite{Sedeno2013} propose another algorithm with the same running time but a reduced memory space requirement of $O(K+m)$.
Both algorithms essentially reduce to solving $O(Km)$ shortest path problems.
This paper proposes an algorithm that finds \emph{all $F$ many optimal integer} $b$-flows, i.e., all minimum cost flows, in time $O(F(m+n)+mn + M)$. We remark that stopping Hamacher's (or Sedeño-Noda's and Espino-Mart\'in's) algorithm at the first sub-optimal flow results in an algorithm with a running time of $O(Fm\bigl(n\log n + m)+M\bigr)$.
Our improvement is essentially made possible by replacing the shortest path sub-problem with a more efficient way to determine a so-called \emph{proper zero cost cycle}  using a modified depth-first search technique.   
As a byproduct, our analysis yields an algorithm for the $K$ best $b$-flow problem that runs in time $O(Kn^3 + M)$.
Our algorithm has the same space requirement as Hamacher's algorithm and may be implemented with the improvement by Sedeño-Noda and Espino-Mart\'in.

Lastly, we remark that other researchers have taken an interest in finding all optimal (or the $K$ best) solutions for combinatorial optimization problems in various applications and refer the reader to the listing in~\cite{Hamacher1995}.

In addition, the determination of the $K$ best flows has a wide range of applications: Dynamic Lot Size problems (see~\cite{Hamacher1995}), multiple objective minimum cost flow problems (see~\cite{Hamacher2007}), or as an instrument to solve minimum cost flow problems with additional constraints (see~\cite{Yan2002}). 
More examples of the practical relevance of determining alternative optimal solutions may be found in~\cite{Galle1989} and~\cite{Baste2019}.

The organization of the paper is as follows. 
In~\Cref{sec:preliminaries}, we briefly introduce the notations and preliminaries. 
\Cref{sec_theoryAndAlgorithm} takes a given minimum cost flow and shows how to obtain another one with a DFS based procedure. \Cref{sec_algoAllOptimal} presents our algorithm for the all optimal flow problem, while~\Cref{improvmentKBest} uses the algorithm to derive a procedure for the $K$ best $b$-flow problem.
We give upper and lower bounds on the number of feasible and optimal flows in~\Cref{UpperBound}.

%Section 2 Preliminaries
\section{Minimum cost flows}
\label{sec:preliminaries}

We assume familiarity with
graph theory basics and state some well-known facts about network flows in the sequel. For details, we refer to \cite{Diestel2006} and \cite{Ahuja1993}, respectively.

Let $D=(V,A)$ be a directed graph with $|V|=n$ nodes and $|A| = m$ arcs 
together with integer-valued, non-negative and finite lower and upper capacity bounds $l_{ij}$ and $u_{ij}$, respectively, for each arc $(i,j)$. 
Moreover, let $b_i$ be the integer-valued \emph{flow balance} of the node $i\in V$. A function $f\colon\, A \rightarrow \mathbb{R}_{\ge0}$ is called a feasible  \emph{flow} if $f$ satisfies 

\begin{align*}
l_{ij} \le f_{ij} \le u_{ij} \quad &\text{ for all } (i,j)\in A, &\text{
	(\emph{capacity constraint}) }\\ 
\sum_{j:(i,j)\in A} f_{ij} - \sum_{j:(j,i)\in A} f_{ji} = b_i 	\quad &\text{ for all } i\in V. &\text{
	(\emph{flow balance constraint})}
\end{align*}

Without loss of generality, we assume that $D$ is connected  (see~\cite{Ahuja1993}) and that there is at least one feasible $b$-flow. 
Let $c\colon\, A \rightarrow \mathbb{R}$  be a function that assigns a \emph{cost} to every arc $(i,j)\in A$.
Then, the \emph{minimum cost flow problem} (MCF problem) asks for a feasible $b$-flow $f$ that minimizes $c(f) := \sum_{(i,j) \in A} c_{ij} f_{ij}$. 
In this case, $f$ is called a \emph{minimum cost flow} or \emph{optimal} flow. 

\begin{definition}
	The \emph{all optimal integer flow problem} (AOF problem) consists of determining all integer flows which are optimal.
\end{definition}

The unimodularity property of the above constraints (see~\cite{Ahuja1993}) guarantees that if there is at least one optimal flow, an optimal integer flow exists if $u$ and $b$ are integral. 
We only consider integer flows for the remainder of this paper.

Let $D_f=(V,A_f:= A^{+}  \cup  A^{-})$ be the \emph{residual graph} with respect to a feasible flow $f$, where 
$
A^+ :=\{(i,j) \mid (i,j)\in A,\, f_{ij} < u_{ij}\}\ \text{ and }  
A^- :=\{ (j,i) \mid (i,j) \in A,\, f_{ij}>l_{ij}\}.
$
In $D_f$, the \emph{residual capacities} and \emph{residual costs} are defined by  $ u_{ij}(f) := u_{ij}-f_{ij} > 0$ and $c_{ij}(f):=c_{ij}$, respectively, if $(i,j)\in A^+$ and $u_{ij}(f):=f_{ji}-l_{ji} > 0$ and $c_{ij}(f):=-c_{ji}$, respectively, if $(i,j)\in A^-$. 

The cost of a directed cycle $C$ in $D_f$ is given by $c(f,C) := \sum_{(i,j) \in C} c_{ij}(f)$ and let $\chi(C) \in \{-1,0,1\}^A$ be the (directed) incidence vector of $C$ with
\begin{align*}
\chi_{ij}(C) := \begin{cases}
1, &\text{if $C$ traverses $(i,j)$ in $D_f$},\\
-1, &\text{if $C$ traverses $(j,i)$ in $D_f$},\\
0, &\text{otherwise}\\
\end{cases} && \text{for all $(i,j) \in A$}.
\end{align*}

We say that $f'$ is obtained from $f$ by \emph{augmenting $f$ along $C$ by $\lambda$} if we set $f' := f + \lambda \chi(C)$. 
The cost $c(f')$ of $f'$ is exactly $c(f) + \lambda c(f,C)$.
An undirected cycle $C_D$ in $D$ is called \emph{augmenting cycle} with respect to a flow $f$ if whenever augmenting a positive amount of flow on the arcs in the cycle, the resulting flow remains feasible. 
Each augmenting cycle with respect to a flow $f$ corresponds to a directed cycle $C$ in $D_f$, and vice versa~\cite{Ahuja1993}. We call $C_D$ the \emph{equivalent} of $C$ in $D$. 

The \emph{negative cycle optimality condition} (see e.g.~\cite{Ahuja1993})  states that $f$ is an optimal flow if and only if $D_f$ cannot contain a negative  cost cycle.

The following fact is well-known for two feasible integer flows (see, for example,~\cite{Schrivjer2003}).

\begin{property} \label{prop:flowDecomposition}
	If $f$ and $f^{\prime}$ are two feasible integer flows, then it holds: $$\begin{aligned}
	& f^{\prime}=f+\sum_{j=1}^k \lambda_j\chi({C_j}), \\
	& c(f^{\prime})=c(f)+\sum_{j=1}^k \lambda_jc(f,C_j)
	\end{aligned}$$
	for directed cycles $C_1,\dots, C_k\subset D_f$ and integers $\lambda_1,\dots \lambda_k>0$.
\end{property}

In particular, two optimal integer flows $f$ and $f^{\prime}$ can only differ on directed  cycles $C_j\subset D_f$ with zero cost.

\begin{definition}
	A \emph{$0$-cycle} is a zero cost cycle $C\subset D_f$, i.e., a cycle $C\subset D_f$ with $c(f,C)=0$.
	We call  a $0$-cycle  $C\subset D_f$ \emph{proper}
	if for all arcs $(i,j) \in C$ we have $(j,i) \not\in C$.
\end{definition}

Then, obtaining a second optimal flow seems to be very simple: Identify a $0$-cycle $C\subset D_f$  and augment $f$ along $C$ by $\lambda$ as defined above, i.e., $f^{\prime}= f + \lambda\chi(C)$
(see,~\Cref{fig3} for an example).

\tikzstyle{vertex}=[circle,fill=black,draw=black,minimum size=8pt,inner sep=0]%
\tikzstyle{edge} = [draw,thick,->]%
\tikzstyle{selected vertex} = [vertex, fill=red!24]%
\tikzstyle{selected edge} = [draw,line width=5pt,-,yellow!50]%
\usetikzlibrary{arrows,automata}%
\pgfdeclarelayer{background}%
\pgfsetlayers{background,main}%

\begin{figure}
	\centering
	\begin{subfigure}{0.3\textwidth}
		\centering
		\begin{tikzpicture}[ ->,shorten >=1pt,auto,node distance=2.8cm,%
		semithick]%
		\foreach \pos/\name in {{(-2,0)/a}, {(-0,0)/b}, {(-2,-2)/c},%
			{(-0,-2)/d}, {(-1,-3.5)/e}}%
		\node[vertex] (\name) at \pos {};%
		\path (a) edge[->,thick]   	node[] {(0,1,8)} (b)%
		(a) edge[->,thick]   	node[sloped,below] {(0,1,3)} (c)%
		(a) edge[->,thick]   	node[sloped] {(3,4,4)} (d)%
		(b) edge[->,thick]   	node[sloped] {(5,5,5)} (d)%
		(c) edge[->,thick]   	node[below] {(0,1,1)} (d)%
		(c) edge[->,thick]   	node[sloped,below] {(2,4,2)} (e)%
		(d) edge[->,thick]   	node[sloped,below] {(2,3,1)} (e);%
		\begin{pgfonlayer}{background}%
		\foreach \source / \dest in {c/d,d/e,c/e}%
		\path[selected edge] (\source.center) -- (\dest.center);%
		\end{pgfonlayer}%
		\node at (-1,-4.5) {$D=(V,A)$};%
		\node at (0.2,-3.5) {$C_D$};%
		\node at (-1,1) {$f$};%
		\end{tikzpicture}%
		\caption{}
	\end{subfigure}%
	\hfill%
	\begin{subfigure}{0.3\textwidth}
		\centering	
		\vspace{0.8cm}
		\begin{tikzpicture}[->,shorten >=1pt,auto,node distance=2.8cm,
		semithick]

		\foreach \pos/\name in {{(0,0)/a}, {(2,0)/b}, {(0,-2)/c},
			{(2,-2)/d}, {(1,-3.5)/e}}
		\node[vertex] (\name) at \pos {};
		
		\path (a) edge[->,thick]   	node[] {8} (b)
		(a) edge[->,thick]   	node[left] {3} (c)
		(a) edge[->,thick,bend right]   	node[left] {4} (d)
		(d) edge[->, thick, bend right]  node[right] {-4} (a)
		(d) edge[->,thick]   	node[right] {-5} (b)
		(c) edge[->,thick]   	node[below] {1} (d)
		(c) edge[->,thick,bend right]   	node[left] {2} (e)
		(e) edge[->,thick,bend right]   	node[left] {-2} (c)
		(e) edge[->,thick,bend left]   	node[right] {-1} (d)
		(d) edge[->,thick,bend left]   	node[] {1} (e);
		
		\begin{pgfonlayer}{background}
		\path  (c) edge[draw,line width=5pt,-,yellow!50] (d)
		
		(d) edge[draw,line width=5pt,-,yellow!50,bend left] (e)
		(e) edge[draw,line width=5pt,-,yellow!50, bend right] (c);
		\end{pgfonlayer}

		\node at (1,-4.5) {$D_f=(V,A_f)$};
		
		\node at (2.5,-3.5) {$C$};

		\end{tikzpicture}%
		\caption{}
	\end{subfigure}%
	\hfill
	\begin{subfigure}{0.3\textwidth}%
		\centering
		\begin{tikzpicture}[ ->,shorten >=1pt,auto,node distance=2.8cm,
		semithick]

		\foreach \pos/\name in {{(4,0)/a}, {(6,0)/b}, {(4,-2)/c},
			{(6,-2)/d}, {(5,-3.5)/e}}
		\node[vertex] (\name) at \pos {};
		
		\path (a) edge[->,thick]   	node[] {(0,1,8)} (b)
		(a) edge[->,thick]   	node[sloped,below] {(0,1,3)} (c)
		(a) edge[->,thick]   	node[sloped] {(3,4,4)} (d)
		(b) edge[->,thick]   	node[sloped] {(5,5,5)} (d)
		(c) edge[->,thick]   	node[below] {(1,1,1)} (d)
		(c) edge[->,thick]   	node[sloped,below] {(1,4,2)} (e)
		(d) edge[->,thick]   	node[sloped,below] {(3,3,1)} (e);

		\node at (5,-4.5) {$D=(V,A)$};
		\node at (5,1) {$f^{\prime}=f+\chi(C)$};
		\end{tikzpicture}%	
		\caption{}
	\end{subfigure}%
	\caption{Example of an augmentation on a 0-cycle with one unit. (a) Graph $D$, the arcs of which are labeled with $(f_{ij}, u_{ij}, c_{ij}) $. Here, the lower bound $l_{ij}=0$ for all arcs $(i,j)\in A$. The 0-cycle $ C_D $ is marked in yellow.  (b)  The residual graph $D_f$ with $ c_{ij} (f) $ as arc labels and the cycle $ C \subset D_f $ in yellow, where $ C_D $ is the equivalent cycle in $D$.  (c) The graph $D$ with new flow values $f^{\prime}= f+\chi(C)$.} \label{fig3}
\end{figure}
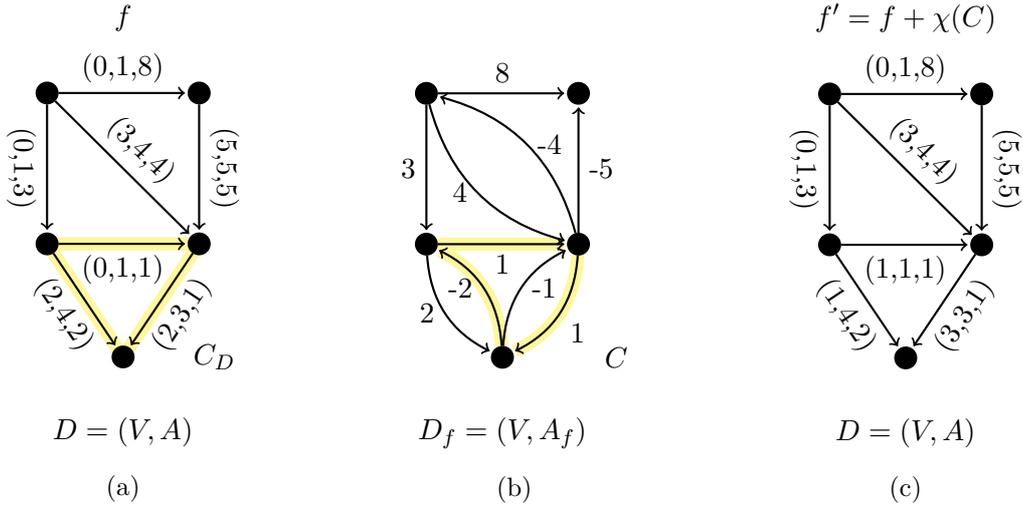

Unfortunately, augmenting $f$ along a $0$-cycle $C$ does not necessarily yield a \emph{different} flow $f'$; for instance, we have $f=f'$ if we augment along $C=( (i,j), (j,i) )$. 
In order to avoid this technical problem, we only consider proper $0$-cycles $C\subset D_f$.

In the exposition of the paper, we assume that $D$ contains no anti-parallel arcs, i.e., if $(i,j)\in A$ then $(j,i)\notin A$, to avoid any confusion with the notation. 
However, the only difficulty with multiple arcs between a pair of nodes is to keep track of the relation that links residual arcs with original arcs, which can be done as mentioned in~\cite{Sedeno2013}. It is straightforward to modify the algorithms in further sections to deal with graphs that contain multiple arcs between a pair of nodes. 
Note that a proper cycle is then defined as a cycle $C$ that does not contain a pair of \emph{symmetric} arcs, i.e., $\{(i,j),(j,i)\}\in D_f$ where both refer to the same arc $(i,j) \in D$.

We associate a real number $y_i$ with each node $i\in V$. We refer to $y_i$ as the node potential of node $i$. For a given set of node potentials $y$
we define the \emph{reduced cost} $\overline{c}_{ij}$ of a given arc $(i,j)$ as $\overline{c}_{ij}:=c_{ij} + y_i - y_j$.
%} 
For forward arcs $(i,j) \in A^+$ of the residual network, the \emph{ residual reduced cost} is defined as $\overline{c}_{ij}(f) := c_{ij}(f) + y_i - y_j = c_{ij} + y_i - y_j$. 
For backward arcs $(j,i) \in A^-$, we define the residual reduced cost as $\overline{c}_{ji}(f) = c_{ji}(f) + y_j - y_i = -c_{ij} + y_j - y_i = -\overline{c}_{ij}$.
\begin{property}\label{prop:cost=reducedCostForCycles}
	The cost  of a cycle $C$   in $D_f$ is 
	$$c(f,C) = \sum_{(i,j)\in C} c_{ij}(f)  =\sum_{(i,j)\in C} c_{ij}(f)+ \sum_{(i,j)\in C}(y_i-y_j) = \sum_{(i,j)\in C} \overline{c}_{ij}(f) = \overline{c}(f,C).$$ 
\end{property}
Therefore, it makes no difference whether we search for cycles in $D_f$ with zero cost or zero reduced cost.

Any node potential $y$ that satisfies the well-known \emph{complementary slackness optimality conditions}  (see e.g.~\cite{Ahuja1993}) is \emph{optimal}, and we have the following property.
\begin{property}\label{prop:complSlack}
	Let $f$ be an optimal flow and $y$ an optimal node potential. Then
	$\overline{c}_{ij}(f)\ge0$  for all $(i,j)\in D_f$.
\end{property}

For an optimal flow $f$, we can determine an optimal node potential $y$ by computing shortest path distances in $D_f$  from an arbitrary but fixed root node $r$ to all other nodes $i\in V$. 
Notice that we assume that all nodes $i$ are reachable from $r$ in $D_f$. If this is not the case, artificial arcs with sufficiently high costs are added to $D_f$. In the following, $y$ is considered as an optimal node potential (see~\cite{Ahuja1993}).

\section{Getting from one Minimum Cost flow to another}
\label{sec_theoryAndAlgorithm}

The idea behind the algorithm presented in~\Cref{sec_algoAllOptimal} is to find proper $0$-cycles efficiently. For this purpose, we will reduce the network $D$. To do so we will use a well-known property about the arcs in a proper $0$-cycle. This result will help us reduce the network and obtain a flow problem in which every feasible solution will yield an optimal solution for the original MCF problem. 

The following property is well-known (see, e.g.,~\cite{Calvete1996,Christofides1986}).
\begin{property}
	\label{prop:differenceOnlyIfCij=0}
	Let $f$ and $f'$ be two optimal integer flows. Then, $f$ and $f'$ only differ on arcs with $\overline{c}_{ij} = 0$. Or, in other words, if $f_{ij} \neq f'_{ij}$  then $\overline{c}_{ij} = 0$.
\end{property}

Let $X :=\{(i,j)\in A\mid \overline{c}_{ij}\neq 0 \}$ be the set of arcs with non-zero reduced cost.
We may remove all arcs in $X$ and obtain a reduced network. 
However, it may be that the arcs in $X$ carry flow. 
Consequently, we have to adjust the demand vector accordingly to maintain a feasible flow. 
Consider the reduced network $D'=(V,A^{\prime})$ with $A^{\prime}=A\backslash X$ and new flow balance values 
$b'_i = b_i - \sum_{(i,j) \in X } f_{ij}+ \sum_{(j,i) \in X} f_{ji}
$. We can easily limit any feasible flow in $D$ to a feasible flow in $D'$ by simply adopting all flow values. 
For a flow $f\in D$, we always refer by $f\in D^{\prime}$ to the reduced flow where all flow values are copied for all arcs $a\in A^{\prime}$. 
We denote by $\mathcal{P}^{\prime}$ the set of all feasible integer flows of the reduced network $D^{\prime}$. 
The fact that there is only a finite number of flows in $\mathcal{P}^{\prime}$ follows from the assumption that there is no arc $(i,j)$ with $u_{ij}=\infty$. 

For a feasible solution $f' \in \mathcal{P'}$, we define
$$ f^*(a)= \left\{ \begin{aligned}
& f^{\prime}(a), & & \text{ if } a\in A^{\prime}, \\ &
f(a), & &  \text{ if } a\in X.
\end{aligned} \right.$$

The next lemma demonstrates how every optimal flow in $D$ can be obtained from the combination of an initial optimal flow and some feasible flow in $D'$.

\begin{lem} \label{lem:feasibleInD'=optimalInD}
	Let $f$ be an optimal integer solution. If $f^{\prime}$ is a feasible integer solution in $D^{\prime}$ and $f^*$ is defined as above, then $f^*$ is an optimal solution in $D$ and vice versa.
\end{lem}
\begin{proof}
	
	\begin{align*}
	&f' \text{ is feasible in }D' & \\
	\iff &f \text{ and } f' \text{ differ on a set of cycles in }D'_{f} & \text{(\Cref{prop:flowDecomposition})}\\
	\iff &f \text{ and } f' \text{ differ on a set of cycles in }D'_{f} &\\
	& \text{ with reduced cost equal to zero} & \text{(Definition of $D'$)}\\
	\iff &f \text{ and } f^* \text{ differ on a set of cycles in }D_{f} &\\
	& \text{with reduced cost equal to zero} & \text{(cycles in $D'_f$ also exist in $D_f$)}\\
	\iff &f^* \text{ is optimal}& \text{(Differ on zero cycles and \Cref{prop:cost=reducedCostForCycles})}
	\end{align*}

\end{proof}

Given a feasible flow $f$ in $D'$, we can compute another feasible flow by finding a proper cycle in the residual network of $D'$. Therefore, we can determine an additional optimal flow in $D$ by finding a proper cycle in $D'_{f}$ instead of finding a proper $0$-cycle in $D_f$. This gives us the following theorem:

\begin{thm}\label{theo:cycle->optimalflow}
	Let $f$ be an optimal integer flow. Then every proper cycle in $D'_f$ induces another optimal integer flow in $D$.
\end{thm}
\begin{proof}
	Any cycle in  $D'_f$ has zero cost in $D_f$.
\end{proof}

We now show that we can determine such a proper cycle with the following depth-first search (DFS) technique. We assume that the reader is familiar with the directed depth-first search. The book by Cormen et al.~\cite{Cormen2001} provides a good overview of this topic.  Nevertheless, we want to repeat a short definition of the different types of arcs in a DFS-forest due to their importance in the next section. 
\begin{definition}[\cite{Cormen2001}]
	\emph{Tree arcs} are the arcs in the trees of the DFS-forest. \emph{Backward arcs} are the arcs $(i,j)$ connecting a vertex $i$ to an ancestor $j$ in the same DFS-tree. Non-tree arcs $(i,j)$ are called \emph{forward arcs} if the arc connects a vertex $i$ to a descendant $j$ in the same DFS-tree. All other arcs are called \emph{cross arcs}.
\end{definition}

In order to find a proper cycle in $D^{\prime}$, let $\mathcal{F}$ be a DFS-forest of $D^{\prime}$. For each tree $\mathcal{T}\in \mathcal{F}$ and each node $i\in \mathcal{T}$, we define $\pi_i$ as the predecessor of $i\in \mathcal{T}$.
Furthermore, let $\operatorname{dfs} (i)$ denote the DFS number (compare \emph{discovery time} in~\cite{Cormen2001}) of node $i$, i.e., the number when node $i$ is visited during the DFS process.  Recall that there cannot exist a cross arc $(i,j)$ with $\operatorname{dfs} (i) < \operatorname{dfs} (j)$, i.e., a cross arc from a node $i$ that is visited before the arc's sink node $j$.  In other words, if we assume that the children of each node are drawn from left to right and the different trees of the DFS-forest are also drawn from left to right, a cross arc always goes from the right to the left. Therefore, there can be cross arcs connecting two different trees of the DFS-forest, but those arcs can never be included in a cycle.
\begin{property}[\cite{Cormen2001}]\label{popertyDFS}
	Each cycle $C$ in $D'_f$ contains at least one backward arc and no cycle can contain an arc connecting two different trees of $\mathcal{F}$, i.e., there is no arc $(i,j)\in C$ with $i\in \mathcal{T}_l$ and $j\in \mathcal{T}_k$ for $\mathcal{T}_l,\mathcal{T}_k\in \mathcal{F}$ with $k\neq l $. 
\end{property}
We call a backward arc $(i,j)$ with $\pi_i\neq j$ a \emph{long backward arc}. Otherwise, we call it a \emph{short backward arc}.
If there exists a long backward arc, we immediately have determined a proper cycle formed by the long backward arc $(i,j)$ and the distinct path from $j$ to $i$ in $\mathcal{T}$ (see~\Cref{fig1}). 
However, even without the existence of such a backward arc, there can exist proper cycles, possibly containing a forward or a cross arc.
See~\Cref{fig2} for examples. However, every proper cycle contains at least one cross, forward or long backward arc.

\tikzstyle{vertex}=[circle,fill=black,draw=black,minimum size=8pt,inner sep=0]
\tikzstyle{edge} = [draw,thick,->]
\tikzstyle{weight} = []
\tikzstyle{selected vertex} = [vertex, fill=red!24]
\tikzstyle{selected edge} = [draw,line width=5pt,-,green!50]
\usetikzlibrary{arrows,automata}
\pgfdeclarelayer{background}
\pgfsetlayers{background,main}
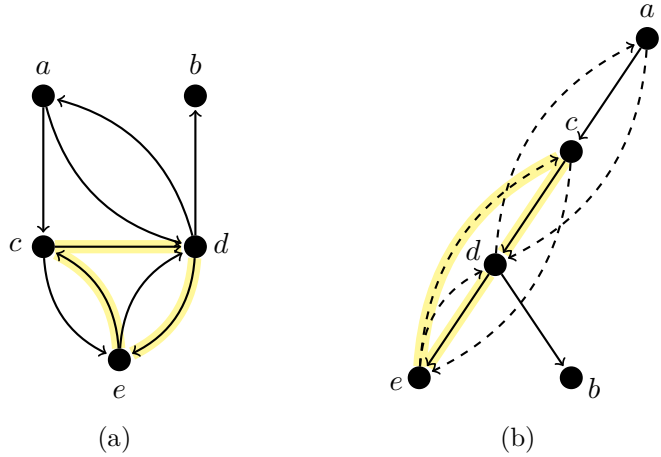
\begin{figure}
	\centering
	\begin{minipage}{0.6\textwidth}
		\subfloat[]{%
			\scalebox{1}{%
				\begin{tikzpicture}[->,shorten >=1pt,auto,node distance=2.8cm,
				semithick]

				\foreach \pos/\name/\direction/\label/\xshift/\yshift in {{(5,0)/a/above/a/0/0}, {(7,0)/b/above/b/0/0}, {(5,-2)/c/left/c/-0.15/-0.12},
					{(7,-2)/d/right/d/0.1/-0.1}, {(6,-3.5)/e/below/e/0/-0.35}}
				\node[vertex] [label={[\direction,xshift=\xshift cm,  yshift=\yshift cm] $\label$}] (\name) at \pos {};
				
				\path 
				(a) edge[->,thick]   	node[left] {} (c)
				(a) edge[->,thick,bend right]   	node[left] {} (d)
				(d) edge[->, thick, bend right]  node[right] {} (a)
				(d) edge[->,thick]   	node[right] {} (b)
				(c) edge[->,thick]   	node[below] {} (d)
				(c) edge[->,thick,bend right]   	node[left] {} (e)
				(e) edge[->,thick,bend right]   	node[left] {} (c)
				(e) edge[->,thick,bend left]   	node[right] {} (d)
				(d) edge[->,thick,bend left]   	node[] {} (e);
				
				\begin{pgfonlayer}{background}
				\path  (c) edge[draw,line width=5pt,-,yellow!50] (d)
				(d) edge[draw,line width=5pt,-,yellow!50,bend left] (e)
				(e) edge[draw,line width=5pt,-,yellow!50, bend right] (c);
				\end{pgfonlayer}

				\end{tikzpicture}}%
		}
		\hfill%
		\subfloat[]{%
			\centering
			\scalebox{1}{\begin{tikzpicture}[->,shorten >=1pt,auto,node distance=2.8cm,
				semithick]

				\node [vertex,label={$a$}] (v1) at (11.5,0) {};
				\node [vertex,label={$c$ }] (v2) at (10.5,-1.5) {};
				\node [vertex,label={[left,xshift=-0.05cm,yshift=0.05cm] $d$}] (v3) at (9.5,-3) {};
				\node [vertex,label={[left,below,xshift=-0.3cm,yshift=-0.4] $e$}] (v4) at (8.5,-4.5) {};
				\node [vertex,label={[right,below,xshift= 0.3cm,yshift=-0.4] $b$}] (v5) at (10.5,-4.5) {};
				
				\draw [edge] (v1) edge (v2);
				\draw [edge] (v2) edge (v3);
				\draw [edge] (v3) edge (v4);
				\draw [edge] (v3) edge (v5);
				\draw [->,thick,bend left,dashed] (v4) edge (v3);
				\draw [->,thick, bend left, dashed] (v4) edge (v2);
				\draw [->,thick, bend left, dashed] (v2) edge (v4);
				\draw [->,thick, bend left, dashed] (v3) edge (v1);
				\draw [->, thick, bend left, dashed] (v1) edge (v3);
				\begin{pgfonlayer}{background}
				\path  (v2) edge[draw,line width=5pt,-,yellow!50] (v3)
				(v3) edge[draw,line width=5pt,-,yellow!50] (v4)
				(v4) edge[draw,line width=5pt,-,yellow!50, bend left] (v2);
				\end{pgfonlayer}
				
				\end{tikzpicture}}
		}
	\end{minipage}
	
	\caption{ (a) The residual graph with a proper cycle marked in yellow and (b) the corresponding DFS tree with the proper backward cycle.}\label{fig1}
\end{figure}
\begin{figure}
	\centering
	\begin{minipage}{0.6\textwidth}
		\subfloat[]{
			
			\scalebox{1}{\begin{tikzpicture}[->,shorten >=1pt,auto,node distance=2.8cm,
				semithick]

				\node [vertex] (v10) at (-5.5,3) {};
				\node [vertex] (v11) at (-6,2) {};
				\node [vertex] (v12) at (-6.5,1) {};
				\node [vertex] (v13) at (-5.5,1) {};
				\node [vertex] (v15) at (-7,0) {};
				\node [vertex] (v16) at (-7.5,-1) {};
				\node [vertex] (v14) at (-6,0) {};
				\draw [edge] (v10) edge (v11);
				\draw [edge] (v11) edge (v12);
				\draw [edge] (v11) edge (v13);
				\draw [edge] (v12) edge (v14);
				\draw [edge] (v12) edge (v15);
				\draw [edge] (v15) edge (v16);
				\draw [edge,dashed, bend left] (v12) edge (v16);
				\draw [edge,dashed,bend left] (v16) edge (v15);
				\draw [edge,dashed,bend left] (v15) edge (v12);
				
				\begin{pgfonlayer}{background}
				\path  (v12) edge[draw,line width=5pt,-,yellow!50,bend left] (v16)
				(v16) edge[draw,line width=5pt,-,yellow!50,bend left] (v15)
				(v15) edge[draw,line width=5pt,-,yellow!50, bend left] (v12);
				\end{pgfonlayer}
				
				\end{tikzpicture}}
		}
		\hfill
		\subfloat[]{
			\centering
			\scalebox{1}{\begin{tikzpicture}[->,shorten >=1pt,auto,node distance=2.8cm,
				semithick]

				\node [vertex] (v1) at (-1,2) {};
				\node [vertex] (v2) at (-1.5,1) {};
				\node [vertex] (v3) at (-2,0) {};
				\node [vertex] (v4) at (-2.5,-1) {};
				\node [vertex] (v5) at (-0.5,1) {};
				\node [vertex] (v6) at (0,0) {};
				\node [vertex] (v7) at (0.5,-1) {};
				\draw [edge] (v1) edge (v2);
				\draw [edge] (v2) edge (v3);
				\draw [edge] (v3) edge (v4);
				\draw [edge] (v1) edge (v5);
				\draw [edge] (v5) edge (v6);
				\draw [edge] (v6) edge (v7);
				\draw [edge,dashed] (v5) edge (v4);
				\draw [edge,dashed, bend left] (v4) edge (v3);
				\draw [edge,dashed,bend left] (v3) edge (v2);
				\draw [edge,dashed,bend left] (v2) edge (v1);

				\begin{pgfonlayer}{background}
				\path  (v4) edge[draw,line width=5pt,-,yellow!50,bend left] (v3)
				(v5) edge[draw,line width=5pt,-,yellow!50] (v4)
				(v3) edge[draw,line width=5pt,-,yellow!50, bend left] (v2)
				(v1) edge[draw,line width=5pt,-,yellow!50] (v5)
				(v2) edge[draw,line width=5pt,-,yellow!50, bend left] (v1);
				
				\end{pgfonlayer}

				\node [vertex] (v8) at (0,3) {};
				\node [vertex] (v9) at (1,2) {};
				
				\draw [edge] (v8) edge (v1);
				\draw [edge] (v8) edge (v9);
				
				\end{tikzpicture}}
		}
	\end{minipage}
	
	\caption{Proper cycles in a tree of $\mathcal{F}$ 
		without containing a long backward arc. (a) A proper forward  cycle and (b) a proper cross cycle.}\label{fig2}
\end{figure}
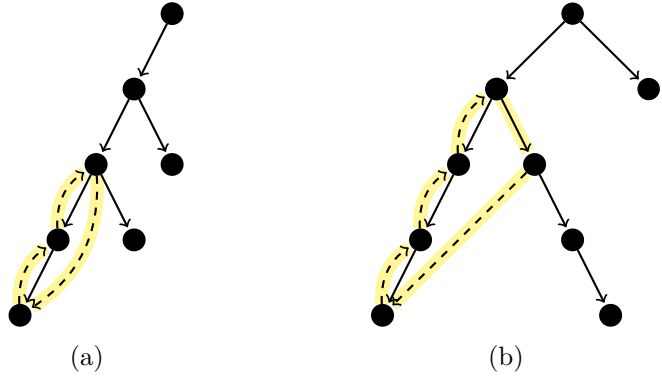

\begin{definition}
	Let $\mathcal{F}$ be a DFS-forest of a directed graph $D$. A \emph{proper backward cycle} in $D$ is a proper cycle which is formed by a long backward arc $(i,j)$  and the unique path from $j$ to $i$ in the corresponding tree $\mathcal{T}$.\\
	A \emph{proper forward cycle} in $D$ is a proper cycle formed by a forward arc $(i,j)$ and the path from $j$ to $i$ only consisting of short backward arcs.\\
	A \emph{proper cross cycle} in $D$ is a proper cycle formed by a cross arc $(i,j)$, the path from $a$ to $i$ in the tree $\mathcal{T}$, where $a$ is the first common ancestor of the nodes $i$ and $j$ in $\mathcal{T}$, and the path from $j$ to $a$ only consisting of short backward arcs.
\end{definition}

\Cref{fig1} and~\Cref{fig2} show examples of the above definitions.
In the next step, we will show that it is enough to search only for these kinds of proper cycles to detect one proper cycle.

\begin{lem}\label{Lem: IfProperThenCycle}
	If there exists a proper cycle in $D^{\prime}_f$, then there also exists a proper backward, forward or cross cycle. 
\end{lem}

\begin{proof}
	Let $C$ be a proper cycle in $D^{\prime}_f$.
	We know that every cycle contains at least one backward arc and cannot contain arcs connecting different trees of the forest.
	\begin{itemize}
		\item[Case 1:] There is at least one long backward arc in $C$. In this case, we have a proper backward cycle.
		
		\item[Case 2:] There is no long backward arc in $C$, but at least one cross arc.
		To simplify, we can replace all forward arcs by tree arcs so that our cycle has only cross arcs, tree arcs and short backward arcs. 
		Now, let $l$ be the node with the least depth in $C$ and let $C$ start with node $l$. Further, let $(i,j)$ be the last used cross arc in $C$. 
		We want to show that a path made from short backward arcs from $j$ to the lowest common ancestor $a$ of $i$ and $j$ exists in the DFS-forest. Then, together with the tree arcs-path from $a$ to $i$, we have a proper cross cycle.
		Since $(i,j)$ is the last cross arc in $C$ before returning to $l$, we have a path from $j$ to $l$ only consisting of short backward arcs.
		Suppose that $a$ has less depth than $l$. It follows that $i$ is in a different sub-tree than the one rooted at $l$.
		To get from $l$ to $i$, cycle $C$ must contain 
		a cross arc leading to the sub-tree of $i$, because in $C$ we cannot change sub-trees by going up the tree above the depth of $l$.
		The definition of cross arcs gives us $\operatorname{dfs} (i) < \operatorname{dfs}(l)$ and $\operatorname{dfs} (j)< \operatorname{dfs}(i) $ and we have $\operatorname{dfs}(l) < \operatorname{dfs}(j)$ because of the backward arcs-path from $j$ to $l$ -- a contradiction.
		So $a$ must have a greater or equal depth than $l$ and therefore is reachable from $j$ by short backward arcs. 
		\item[Case 3:] The cycle $C$ contains no long backward arcs, no cross arcs but at least one forward arc. The only other kinds of arcs left are tree arcs and short backward arcs. Cycle $C$ must contain short backward arcs, especially the ones from the sink of a forward arc to its source to close the cycle. Thus, we have at least one proper forward cycle. In fact, $C$ cannot contain tree arcs since it is proper and we have no long backward arcs.
	\end{itemize}
\end{proof}

\section{The All Optimal Integer Flow Algorithm}
\label{sec_algoAllOptimal}

This section presents an algorithm that, given an initial optimal integer solution $f$ to the MCF, determines all $F$ many optimal integer flows in $\mathcal{O}(F(m+n)+mn)$ time. 
First, we specify how to find an optimal flow different from the given one. 
Then, we show how we can partition the solution space of the AOF problem in $D$, given two optimal solutions. This gives us two smaller AOF problems and, recursively, we can find all optimal integer flows.

First, we provide an algorithm to determine a second feasible flow by determining one proper cycle of the aforementioned kinds. Herefore, we define a variable \emph{$\textit{SBAlow(i)}$}
for each node $i\in V$ in the DFS-forest defined as the dfs number of the  node closest to the root reachable from $i$ only using short backward arcs. We may compute the $\textit{SBAlow}$ value analogously to the lowpoint in~\cite{HT1973} during the DFS-procedure: Whenever the DFS visits a node  $i$, we check whether there is a short backward arc $(i,\pi_i)$. If so, we set $\textit{SBAlow(i)}= \textit{SBAlow}(\pi_i)$; otherwise, $\textit{SBAlow(i)}= \text{dfs}(i)$.
With these variables, we can quickly check if a path exists that  only consists of short backward arcs from one node to another.
This will help us to decide if proper forward or cross  cycles exist. 	

We use~\Cref{Lem: IfProperThenCycle}, which guarantees that if any proper cycle exists, there is either a proper forward, backward or cross cycle. Therefore, it suffices to check for these three kinds of cycles only.
If one long backward arc exists, we have already determined a proper backward cycle. For any forward arc $(i,j)$, we test if a path from $j$ to $i$ exists only containing short backward arcs by checking if $\textit{SBAlow}(j) \le \text{dfs}(i)$. In this case, we know that there exists a path from $j$ to $i$ only consisting of short backward arcs, and hence we have determined a proper forward cycle. 
For any cross arc $(i,j)$ we have to test if $\textit{SBAlow}(j) \le \text{dfs}(a)$, where $a$ is the first common ancestor of $i$ and $j$. If so, we know that a path from $j$ to $a$ exists only containing short backward arcs. 
In this case, we have determined a proper cross cycle.
The following~\Cref{algo:AnotherFeasible}  determines a proper cycle if one exists and in this case returns a second feasible solution.

\begin{algorithm}
	\KwData{feasible flow $f$ in $D$}
	\KwResult{Another feasible flow $f^{\prime}$ with $f^{\prime}\neq f $, if one exists.}
	\BlankLine
	$D_f \leftarrow$ BuildResidualNetwork($D$)\;
	$\mathcal{F} \leftarrow$ BuildDFSForest($D_f$) \tcp*{Also determines all $\textit{SBAlow}$ values}
	Preprocess($\mathcal{F}$)\tcp*{Determine forward, backward and cross arcs}
	%\For{ all trees $\mathcal{T} \in \mathcal{F}$}
	%{ 
	
	\For{ all backward arcs $(i,j)$}
	{
		\If{ $\pi_i\neq j$}
		{
			\tcp{$(i,j)$ is long backward arc}	
			\tcp{Build cycle from backward arc and $j$-$i$ tree path}			
			$C_{ij} \leftarrow$ proper backward cycle $(i,j) +j$-$i$ tree arcs path\;
			$f^{\prime} \leftarrow f+\chi({C_{ij}})$\;
			\KwRet{$f'$}\;
		}
	}
	
	\For{all forward arcs $(i,j)$}
	{
		
		\If{ $\textit{SBAlow}(j) \le $ dfs$(i)$}
		{
			\tcp{\fontdimen2\font=0.8ex Build cycle from forward arc and $j$-$i$ short backward arcs~path }
			$C_{ij} \leftarrow $ proper forward cycle $(i,j)+P_{ji}$\;
			$f^{\prime} \leftarrow f+\chi({C_{ij}})$\;			
			\KwRet{$f'$\;}
		}
		
	}

	\For{all cross arcs $(i,j)$}
	{
		$a\leftarrow$ LowestCommonAncestor$(i,j)$ \;
		\If{ $\textit{SBAlow}(j) \le \text{dfs}(a)$}
		{
			\tcp{Build cycle from cross arc and $j$-$i$ short backward arcs path}
			$C_{ij} \leftarrow $ proper cross cycle $(i,j)+P_{ja} + a\text{-}i$ tree arcs path\;
			$f^{\prime} \leftarrow f+\chi({C_{ij}})$\;			
			\KwRet{$f'$\;}
		}			
		
	} 
	%}
	\tcp{No cycle found and therefore no other flow}
	\KwRet{\Null\;}
	
	\caption{FindAnotherFeasibleFlow }\label{algo:AnotherFeasible}
\end{algorithm}

\begin{lem}\label{lem:ProofForAlgoAnotherFeasible}
	Given a feasible integer flow $f$ in $D$,~\Cref{algo:AnotherFeasible} determines another feasible integer flow different from $f$ in time $\mathcal{O}(m+n)$ or decides that no such flow exists. 
\end{lem}

\begin{proof}
	
	The algorithm finds one proper backward, forward or cross cycle, if one exists, and therewith constructs a feasible flow distinct from $f$. The correctness follows from~\Cref{theo:cycle->optimalflow} and~\Cref{Lem: IfProperThenCycle}.
	
	We can create the residual network and a DFS-forest in time $\mathcal{O}(m+n)$.
	Determining the kinds of arcs (tree, forward, backward or cross arcs) can also be accomplished in  $\mathcal{O}(m+n)$ time:
	During the computation of the DFS-forest, we mark all tree arcs. If for a non-tree arc $(i,j)$ we have $\operatorname{dfs}(i) <  \operatorname{dfs} (j)$, then $(i,j)$ is a forward arc. Otherwise, we check the lowest common ancestor of $i$ and $j$. If it is $j$, we have a backward arc, otherwise $(i,j)$ is a cross arc. 
	The lowest common ancestor queries can be answered in constant time provided a linear preprocessing time~\cite{HarelT84}. The determination of the $\textit{SBAlow}$ values can also be accomplished during the DFS-procedure in $\mathcal{O}(m+n)$ time.
	
	Let $m= m_1 \thinspace \dot{\cup} \thinspace m_2 \thinspace \dot{\cup} \thinspace m_3 \thinspace \dot{\cup} \thinspace m_4$ be a partition of $m$, where $m_1$ are all tree arcs, $m_2$ all backward arcs, $m_3$ all forward arcs and $m_4$ all cross arcs in the DFS forest. 
	We test for all backward arcs $(i,j)$ in $\mathcal{O}(1)$ time whether $\pi_i \neq j$. In this case, we encounter a long backward arc and we need additional $\mathcal{O}(n)$ time to create the proper backward cycle $C_{ij}$ and adapt the flow. So, we need $\mathcal{O}(m_2)$ time for all backward arcs.
	For a forward arc, we also test in time $\mathcal{O}(1)$ if $\textit{SBAlow}(j) \le \text{dfs}(i)$. Overall, we need $\mathcal{O}(m_3)$ time for all forward arcs. As discussed before, the construction of a possible proper forward cycle and the returned flow take $\mathcal{O}(n)$ time.
	We also  check if $\textit{SBAlow} \le \text{dfs}(a)$ in $\mathcal{O}(1)$ time for cross edges. The lowest common ancestor $a$ is determined in $\mathcal{O}(1)$ time. Again, a proper cross cycle and the  flow $f'$ can be created in $\mathcal{O}(n)$ time.
	Altogether, we have a run time of $\mathcal{O}(m_2) +\mathcal{O}(m_3) + \mathcal{O}(m_4) + \mathcal{O}(n)$. The last term corresponds to the construction of $C_{ij}$ and $f'$.
	Thus, the algorithm determines another feasible integer flow in $\mathcal{O}(n+m)$ or decides that none exists. \qedhere
	
\end{proof}

However, there is no need to first build the DFS-forest and then test if a proper cycle exists for the different arc sets. We may  also integrate all checks during the DFS-procedure, ending with a probably faster algorithm in practice. We choose the presented formulation of~\Cref{algo:AnotherFeasible} for better readability.

Next we describe how we use~\Cref{algo:AnotherFeasible} to determine another optimal flow in a network.
Due to~\Cref{lem:feasibleInD'=optimalInD}, any feasible flow in the reduced network $D'$ is an optimal flow in $D$. Therefore, we only need to reduce the network and compute a second feasible flow to determine another optimal flow different from an initial optimal flow. \Cref{algo:AnotherOptimal} follows this procedure.

\begin{algorithm}[t]
	\KwData{optimal flow $f$ in $D$ and the reduced costs $\overline{c}$}
	\KwResult{Another optimal flow $f^{*}$ with $f^{*}\neq f $, if one exists.}
	\BlankLine
	\tcp{Reduce network}
	$D' \leftarrow D$\;
	$A'$ $\leftarrow$ $\{(i,j) \in A \;|\; \overline{c}_{ij} = 0\}$\;
	$b' \leftarrow$ AdaptDemandVector($f$, $D'$)\;
	\BlankLine
	$f' \leftarrow$ FindAnotherFeasibleFlow($f$, $D'$)\;
	\lIf{f' = \Null}{Stop}
	\For{all arcs $a\in D$}
	{\lIf{$a \in A'$}{$f^*(a) \leftarrow f'(a) $}
		\lElse{$f^*(a) \leftarrow f(a) $}
	}	
	\KwRet{$f^*$\;}
	\caption{FindAnotherOptimalFlow}\label{algo:AnotherOptimal}
\end{algorithm}

\begin{lem} \label{lem: RuntimeAnotherOptimalFlow}
	Given an initial optimal integer flow $f$ in $D$ and the reduced costs $\overline{c}$,~\Cref{algo:AnotherOptimal} determines another optimal integer flow different from $f$ in $\mathcal{O}(m+n)$ time or decides that no such solution exists.
\end{lem}
\begin{proof}
	The algorithm reduces the network $D$ following the description from the paragraph prior to~\Cref{lem:feasibleInD'=optimalInD}.  Thereafter, it computes another feasible flow $f'$ in the reduced network given an optimal flow $f$ in $D$, which is naturally feasible in $D'$.
	Then, $f'$ can be transformed into another optimal flow $f^*$ in $D$.
	The correctness of~\Cref{algo:AnotherOptimal} follows from~\Cref{lem:feasibleInD'=optimalInD}.

	We can reduce the network in $\mathcal{O}(m)$ time, and due to~\Cref{lem:ProofForAlgoAnotherFeasible}, determine another feasible flow  in the reduced network in $\mathcal{O}(m^{\prime}+n)$ time using the initial flow $f$ as input. Here, $m'$ denotes the number of arcs in $D'$. Therefore, we can find another optimal solution in time $\mathcal{O}(m+n)$ or decide that none exists.
\end{proof}

Since we can compute a second optimal flow in $D$, we can now apply the binary partition approach as used in~\cite{Hamacher1995} or~\cite{Sedeno2013}. 
Let $\mathcal{P}$ be the set of optimal flows in $D$, which we also call the \emph{optimal solution space}, and let $f_1,f_2 \in \mathcal{P}$.~\Cref{algo:FindAllOptimalFlows} determines all optimal flows, by recursively dividing $\mathcal{P}$ into $\mathcal{P}_1$ and $\mathcal{P}_2$ such that $f_1$ and $f_2$ are optimal flows in $\mathcal{P}_1$ and $\mathcal{P}_2$, respectively, and then seeks more flows in $\mathcal{P}_1$ and $\mathcal{P}_2$. 
This can easily be done by identifying some arc $(i,j)$ with $f_{1,ij}\neq f_{2,ij}$. Such an arc must exist since $f_1\neq f_2$. 
Assume w.l.o.g. that $f_{1,ij} < f_{2,ij}$. We set  
$$\mathcal{P}_1:=\{z\in \mathcal{P}: z_{ij}\le f_{1,ij}\}$$ and $$\mathcal{P}_2:=\{z\in \mathcal{P}: z_{ij}\ge f_{1,ij}+1\}.$$

\begin{algorithm}[t]
	\KwData{optimal flow $f$ in $D$ and the reduced costs $\overline{c}$}
	\KwResult{All optimal integer flows in $D$}
	\BlankLine
	$f^* \leftarrow$ FindAnotherOptimalFlow($f$, $\overline{c}$, $D$)\;
	\lIf{$f^*=$ \Null}{Stop}
	\BlankLine
	Print($f^*$)\;
	\tcp{Partition solution space and find new optimal flow}
	$a \leftarrow$ arc $a$ with $f(a) \neq f^*(a)$\;
	\eIf{$f(a) < f^*(a)$}
	{FindAllOptimalFlows($f$, $\overline{c}$, $D$) with $u_a=f(a)$\;
		FindAllOptimalFlows($f^*$, $\overline{c}$, $D$) with $l_a=f(a)+1$\;
	}
	{FindAllOptimalFlows($f$, $\overline{c}$, $D$) with $l_a=f(a)$\;
		FindAllOptimalFlows($f^*$, $\overline{c}$, $D$) with $u_a=f(a)-1$\;
	}	
	\caption{FindAllOptimalFlows}\label{algo:FindAllOptimalFlows}
\end{algorithm}

This partitioning implies that each of the new minimum cost flow problems is defined in a modified flow network with an altered upper or lower capacity of a single arc $(i,j)$. By dividing the new solution spaces again until no more other optimal flow exists, all optimal solutions are found, whose number we denote by $F$.  
Our main result follows:

\begin{thm}\label{them: allOptimalFlows}
	
	Given an initial optimal integer solution $f$ and the reduced costs $\overline{c}$,~\Cref{algo:FindAllOptimalFlows}  solves the AOF problem in $\mathcal{O}(F(m+n))$ time, where $F$ is the number of all optimal integer flows.
\end{thm}

\begin{proof}
	In each recursive call of~\Cref{algo:FindAllOptimalFlows}, a new optimal integer flow  different from the input one is found if one exists. 
	Now let $\mathcal{P}_i$ be the current optimal solution space in an algorithm call and let $\mathcal{P}_j$ and $\mathcal{P}_k$ be the optimal solution spaces of the two following calls.
	Every optimal flow is found since $\mathcal{P}_i = \mathcal{P}_j \cup \mathcal{P}_k$ and no flow is found twice because $\mathcal{P}_j \cap \mathcal{P}_k = \emptyset$.
	
	We will now show that we need $\mathcal{O}(F)$ calls of~\Cref{algo:FindAllOptimalFlows} to find all optimal integer flows. We associate each algorithm call with one optimal flow. 
	If we find another optimal solution, we associate this call with the newly found flow. Otherwise, we associate the call with the input flow.
	Because each flow induces exactly two more calls, there are at most three calls associated with one optimal flow. 
	The primary time effort of~\Cref{algo:FindAllOptimalFlows} is the computation of another optimal flow (\Cref{algo:AnotherOptimal}), which takes $\mathcal{O}(m+n)$ time. 
	We do not have to determine the reduced cost again, since  changing the flow value on arcs with zero reduced cost maintains their reduced costs.
	So in total, we solve the AOF problem in $\mathcal{O}(F(m+n))$  time.
\end{proof}

\subsection{Remarks}
The initial optimal integer flow may be obtained by using, for example, the network simplex algorithm or the enhanced capacity scaling algorithm (\cite{Ahuja1993},\cite{Orlin1993}). 
The latter determines an optimal flow in $\mathcal{O}((m\log n)(m+n\log n))$, i.e., in strongly polynomial time~\cite{Orlin1993}.

Since we can compute the node potential as described in~\Cref{sec:preliminaries} by a shortest path problem in $D_f$, possibly
containing arcs with negative cost, we need $\mathcal{O}(mn)$ time to determine the nodes' potentials. 
Given these potentials, we can determine the reduced costs in $\mathcal{O}(m)$ time. Then, an algorithm for the determination of all $F$ many integer flows needs $\mathcal{O}(F(m+n)+mn)$ time. 

Using the network-simplex algorithm yields a so-called basic feasible solution in which an optimal node potential can be obtained in $\mathcal{O}(n)$ time (see ~\cite{Ahuja1993}). Then, an algorithm for the determination of all $F$ many integer flows reduces to $\mathcal{O}(F(m+n))$ time.  
Sedeño-Noda and Espino-Martín introduce techniques in~\cite{Sedeno2013}  to reduce the memory space of the partitioning algorithm. With these techniques, it is possible to implement~\Cref{algo:FindAllOptimalFlows} to run in space $\mathcal{O}(F+m)$. 
Furthermore, it is not necessary to reduce the graph at each call of~\Cref{algo:AnotherOptimal}. 
Instead, it is sufficient to reduce the network $D$ once in the beginning and to determine all feasible solutions in the same network $D^{\prime}$ because arcs with zero reduced cost maintain their reduced costs. 
However, this does not affect the asymptotical running time. 
We shall use this implementation of~\Cref{algo:AnotherOptimal} in the next chapter.

The algorithm may also be used to determine all circulations, all feasible flows and all $s$-$t$ flows with given value $\phi$ since these flow problems may be transformed into each other.

%Section 5 Improvment K-Best Flow Algo
\section{Improved running time for the $K$ best flow problem}
\label{improvmentKBest}

In the present section, we will improve the $K$ best flow algorithm by Hamacher~\cite{Hamacher1995} and derive a new one with running time $\mathcal{O}(Kn^3)$ or $\mathcal{O}(K(n^2\log n +mn))$ for sparse graphs, instead of $\mathcal{O}(K(mn\log n +m^2))$.

\begin{definition}[ \cite{Sedeno2013} ]
	\emph{The $K$ best integer flow (KBF)} problem consists of determining the $K$ best integer solutions of the MCF problem. In other words, identifying $K$ different integer flows $f_k$ with $k\in \{1,\dots, K \}$ such that $c(f_1)\le c(f_2)\le \dots \le c(f_k)$ and there is no other integer flow $f_p\neq f_k$ with $k\in \{1,\dots, K\}$ with $c(f_p)<c(f_k)$. 
\end{definition}

As in Hamacher's algorithm~\cite{Hamacher1995}, we want to determine a second optimal flow by determining a proper cycle of minimal cost. 
Using results from the previous sections, we can find this cycle more efficiently.
Recall that we obtain a second-best flow by augmenting an optimal flow over a proper minimal cycle with one flow unit. After this, we will use the binary partition approach, as presented in~\cite{Hamacher1995} or~\cite{Sedeno2013}. 

Hamacher introduces an idea to improve the average complexity of his algorithm:
Let $\overline{D}=(\overline{d}_{ji})$ be the distance table of $D_f$ concerning $\overline{c}(f)$, i.e., $\overline{d}_{ji}$ is the length of the shortest path $P_{ji}$ from $j$ to $i$ in $D_f$ with length $\overline{c}(f,P_{ji})$. 
The distance table may be computed in time $\mathcal{O}(n^3)$ using the Floyd-Warshall algorithm or in time $\mathcal{O}(n^2\log n+mn)$ by (essentially) repeated calls to Dijkstra's algorithm; the latter is more efficient on sparse graphs.

\begin{definition}[\cite{Hamacher1995}]
	We define a \emph{proper $(i,j)$-cycle } $C\in D_f$  as a proper cycle that contains the arc $(i,j)$. In addition, a proper  $(i,j)$-cycle $C\in D_f$ is called \emph{minimal} if it has minimal cost overall proper ($i,j)$-cycles. 
\end{definition}

Hamacher shows the following property for proper minimal $(i,j)$-cycles:

\begin{property}[\cite{Hamacher1995}]
	\label{prop:(i,j)min-cylce}
	For any arc $(i,j)\in A_f$ with  $(j,i)\notin A_f$, i.e., for any arc with no anti-parallel arc in $A_f$, the cost of a proper minimal $(i,j)$-cycle $C\in D_f$ is given by $c(f,C)=\overline{c}_{ij}(f)+\overline{d}_{ji}$.
\end{property} 

When using Hamacher's idea, the problem is that it only applies to arcs with no anti-parallel arc in $A_f$. For all other arcs, the cost of the corresponding proper minimal $(i,j)$-cycle has to be computed by finding a shortest path in $D_f \setminus \{(j,i)\}$.
In the following result, we show that we can use the  previously obtained~\Cref{algo:FindAllOptimalFlows} and the complementary slackness optimality conditions of an optimal solution to overcome this problem and restrict ourselves to consider only arcs with no anti-parallel arcs in $D_f$.

\begin{algorithm}[t]
	\KwData{optimal flow $f$ in $D$}
	\KwResult{A second-best flow $f^{\prime}$ with $f^{\prime}\neq f $, if one exists.}
	\BlankLine
	
	\BlankLine
	$y$ $\leftarrow$ ComputeNodePotential($f$, $D$)\;
	$\overline{c}$ $\leftarrow$ ComputeReducedCost($y$, $D$)\;
	$f' \leftarrow$ FindAnotherOptimalFlow($f$, $\overline{c}$, $D$)\;
	\If{$f' = \Null $}{
		$(\overline{D},P) \leftarrow$ DetermineDistanceTableAndPaths($\overline{c},f,D$)\;
		$\overline{A} \leftarrow \{(i,j)\in D_f \mid  (i,j)\in D \text{ with }  f_{ij}=l_{ij} \text{ or } (j,i)\in D \text{ with }  f_{ji}=u_{ji}\}$\;
		$C=\arg\min\{c(f,C)=\overline{c}_{ij} + \overline{d}_{ji} \mid  C=\{(i,j)\}\cup P_{ji} \text{ with } (i,j)\in \overline{A}\} $\;
		\lIf{$C = \Null$} {Stop}
		$f'\leftarrow f + \chi(C)$\;} 
	\BlankLine
	\KwRet{$f'$\;}
	\caption{FindSecondBestFlow}\label{algo:SecondBest}
\end{algorithm}

\begin{lem}\label{lem: runtimeSecondBest}
	Given an initial optimal integer flow $f$,~\Cref{algo:SecondBest}  determines a second-best flow in $\mathcal{O}(n^3)$ time or decides that no such flow exists.
\end{lem}

\begin{proof}
	
	Let $f$ be an optimal integer solution. Assume that there exists another optimal solution $f^{\prime}$ in $D$. Then $D_f$ contains at least one proper $0$-cycle. Since we can compute one proper $0$-cycle with~\Cref{algo:AnotherOptimal} by~\Cref{lem: RuntimeAnotherOptimalFlow}, we can determine a second-best integer flow.
	
	So assume that there exists no other optimal flow different from $f$.  Let $C$ be a minimal proper cycle. No proper $0$-cycle  can exist in $D_f$ due to  \Cref{theo:cycle->optimalflow}. 
	So $\overline{c}(C,f)>0$, because
	there cannot be a cycle with negative costs due to the negative cycle optimality conditions. 
	Therefore, there exists an arc $(i,j)\in C$ with $\overline{c}_{ij}(f)>0$. Since $f$ is an optimal solution, the complementary slackness optimality conditions give us that $\overline{c}_{ij}(f)> 0$ only holds if $(i,j)\in D $ and $f_{ij}=l_{ij}$ or $(j,i)\in D$ and $f_{ji}=u_{ji}$. 
	Therefore, we have that $(i,j)\in \overline{A}:=\{(i,j)\in D_f \mid  (i,j)\in D \text{ with }  f_{ij}=l_{ij} \text{ or } (j,i)\in D \text{ with }  f_{ji}=u_{ji}\}$ and $(j,i)\notin D_f$, (since the flow value of the corresponding arc of $(i,j)$ in $D$ is equal to the upper or lower capacity of this arc). 
	Since $(j,i)\notin D_f$ it holds that $\overline{c}(C,f)= \overline{c}_{ij}+ \overline{d}_{ji}$ due to \color{blue}Property\color{black}~\ref{prop:(i,j)min-cylce}. So~\Cref{algo:SecondBest} finds a minimal proper cycle and hence a second-best optimal flow. 
	
	Now for the run time.	
	Computing the node potentials by solving a shortest path problem in $D_f$ as described in~\Cref{sec:preliminaries} takes $\mathcal{O}(mn)$ time because of possible negative arc cost.
	Given these potentials, we can determine the reduced costs in $\mathcal{O}(m)$ time. If there exists a proper $0$-cycle, we can determine one in time $\mathcal{O}(m+n)$ by~\Cref{lem: RuntimeAnotherOptimalFlow} or decide that no proper $0$-cycle exists.
	After that, we compute the distance table in $\mathcal{O}(n^3)$ time.  Notice that this approach also provides the paths $P_{ji}$ in addition to the cost $\overline{d}_{ji}$.
	The minimal proper cycle is given by $\arg\min\{c(f,C)\mid C=\{(i,j)\}\cup P_{ji} \mid (i,j)\in \overline{A} \} $.
	Given the distances $\overline{d}_{ji}$ and the reduced costs $\overline{c}_{ij}(f)$ for all $(i,j)$ in $D_f$ we can determine $c(f,C)=\overline{c}_{ij} + \overline{d}_{ji}$ in $\mathcal{O}(1)$ and can determine the $\arg\min$ in $\mathcal{O}(m)$ time.

	As a result, we can compute a minimal proper cycle and, therefore a second-best integer flow in time $\mathcal{O}(n^3)$.
\end{proof}

\begin{algorithm}[t]
	\KwData{optimal flow $f$ in $D$}
	\KwResult{$K$ best flows }
	\BlankLine
	
	$k \leftarrow 1$\;
	$f' \leftarrow$ FindSecondBestFlow($f$, $D$)\;
	\tcp{Insert $f$ and $f'$ in binary heap with $c(f')$ as key}
	\lIf{$f'$ $\neq$ \Null}{
		$B \leftarrow B \cup \{(f,f',D)\}$}
	\While{$B\neq \emptyset $ and $k \leq K$ }{
		$k++$\;
		$(f^1_k,f^2_k, D_k) \leftarrow ExtractMin(B)$\;
		$Print(f^2_k)$\;
		\lIf{$k=K$}{\Break}
		\tcp{Partition solution space by adjusting upper and lower bounds so that $f^1_k$ is optimal flow in $D^1_k$ and $f^2_k$ is optimal flow in $D^2_k$ (compare~\Cref{algo:FindAllOptimalFlows})}
		$\{D^1_k,D^2_k\}\leftarrow Partition(f^1_k,f^2_k,D_k)$\; 
		$f'\leftarrow $ findSecondBestFlow$(f^1_k,D^1_k)$\;
		\lIf{$f'$ $\neq$ \Null}{$B\leftarrow B\cup \{f^1_k,f',D^1_k\}$}
		$f''\leftarrow $ findSecondBestFlow$(f^2_k,D^2_k)$\;
		\If{$f''$ $\neq$ \Null}{$B\leftarrow B\cup \{f^2_k,f'',D^2_k\}$}
	}
	\BlankLine
	\caption{FindKBestFlows}\label{algo:KBest_v2}
\end{algorithm}

Since we are able to find a second-best flow in $\mathcal{O}(n^3)$ time and by using the binary partition approach as in~\cite{Hamacher1995} or~\cite{Sedeno2013},  we obtain the following result:

\begin{thm}
	Given an optimal integer flow $f$,~\Cref{algo:KBest_v2} determines the $K$ best integer solutions for a given MCF problem in time $\mathcal{O}(Kn^3)$. %and $\mathcal{O}(K+m)$ memory space.
\end{thm}

\begin{proof}
	The while loop iterates at most $K-1$ times.
	Due to~\Cref{lem: runtimeSecondBest}, the second-best flow determination takes $\mathcal{O}(n^3)$ time in each iteration. 
	The binary heap insertion and the extraction operations require $\mathcal{O}(m\log m)$ time~\cite{Sedeno2013}.
	So the algorithm needs $\mathcal{O}(Kn^3)$ time overall. 
	The uniqueness of any flow follows from the used partitioning approach~\cite{Hamacher2007}.
\end{proof}

\section{Bounds on the number of feasible and optimal flows}
\label{UpperBound}

In this section, we will introduce further theoretical results. For better readability, we moved some of the proofs to the appendix since the ideas and statements might also be intuitively comprehensible.
We then use these results to give an upper and a lower bound for the number of all optimal and of all feasible integer flows. The main characteristics of the obtained bounds  bases on the fact that any optimal solution can be obtained by an initial optimal \emph{tree-solution} plus a conical combination of all incidence vectors of the \emph{zero cost induced cycles} with bounded coefficients.  Again, we only consider integer flows for the remaining of this paper.

We start with some well-known notations.
Let us call the set of all feasible flows in $D$ the \emph{flow polyhedron} $P(D,l,u,b).$
A feasible flow $f$ is a \emph{basic feasible} solution if it corresponds to a vertex of the flow polyhedron $P(D,l,u,b)$.
Any basic feasible solution $f$ of (MCF) corresponds to a (not necessarily unique) spanning-tree structure $(T,L,U)$. For simplicity, we shall abbreviate \emph{spanning tree} as \emph{tree} in the sequel and refer to a tree $T$ when we really mean its set of arcs.
\begin{definition}
	We call $f$ and an associated \emph{tree structure} $(T,L,U)$ a \emph{tree solution} if it consists of  a spanning tree $T$ of $D$ and a disjoint set $A\backslash T= L \, \dot{\cup} \, U$ such that the flow $f$ satisfies
	$$\begin{aligned} &f_{ij} = l_{ij}  & & \text { for all } (i,j) \in L, \\ &f_{ij} =u_{ij} & & \text { for all } (i,j) \in U. \end{aligned}$$ 
	
\end{definition}

If a feasible (optimal) flow exists, then there also exists a feasible (optimal) tree solution~\cite{Cook1998}, respectively. 
The network simplex algorithm~\cite{Dantzig1951} always determines an optimal tree solution~\cite{Ahuja1993}. We assume w.l.o.g that $l_{ij} = 0$ for all $(i,j)\in A$ (see~\cite{Ahuja1993}).

We stated above that a flow $f$ is optimal if and only if $D_f$ does not contain a negative cycle and
a tree solution helps us find cycles in $D_f$ quickly:
Indeed, if we have a tree solution $f$ with an associated tree structure $(T,L,U)$, it suffices to check cycles that arise by inserting a single edge of $D\setminus T$ into $T$ in order to prove optimality of $f$, see~\cite{Cook1998} for details. 
This check can be further simplified if we consider a node potential, i.e., a vector $y \in \mathbb{R}^V$, where $y_v$ is defined to be the cost of the (unique) simple undirected path in $T$ from an arbitrary but fixed node $r$ to $v$.  
Here, the cost of an arc $(i,j) \in D$ traversed in forward direction by the shortest path contributes positively to the total cost, whereas the cost of a backwards-traversed arc contributes negatively.
Then, the following property holds:

\begin{property}[\cite{Ahuja1993}]
	A tree solution $f$ with an associated tree structure $(T,L,U)$ is optimal if 
	$$ \begin{aligned}
	&(i) & & \overline{c}_{ij} = 0 & & \text{ for all } (i,j)\in T, \\
	&(ii) & & \overline{c}_{ij} \ge 0 & & \text{ for all } (i,j)\in L, \\
	&(iii) & &  \overline{c}_{ij} \le 0 & & \text{ for all } (i,j)\in U.
	\end{aligned}$$
\end{property}

These three conditions are equivalent to $\overline{c}_{ij}(f) \geq 0$ for all $(i,j) \in D_f$.
The equation $(i)$ is always true by our definition of the node potential $y$. Together with conditions $(ii)$ and $(iii)$, it follows that any cycle in $D_f$ has a non-negative cost and hence $f$ is an optimal flow.

The overall advantage of this kind of node potential is that we can construct cycles by inserting one single edge to $T$ with  some nice properties, which we will introduce below.

\begin{definition}[\cite{Cook1998}]\label{def:Ca}
	Let $(i,j) \notin T$ be a non-tree arc.
	\begin{itemize}
		\item[(i)]	There exists a \emph{unique} cycle $C(T, L, U, (i,j))$  \emph{induced by }$(i,j)$, that is formed by $(i,j)$ together with the $j-i$ path in $T$,  designated with $P^T_{j,i}$. This cycle is from now on referred to  as $C_{ij}$.
		\item [(ii)]  The arc $(i,j)$ defines the orientation of the cycle $C_{ij}$. If $(i,j)\in L$, then the orientation of the cycle is in the same direction as $(i,j)$. If otherwise $(i,j)\in U$, then the cycle's orientation is in the opposite direction. We define the set of all arcs $(u,v)$ in $C_{ij}$ that are in the same direction as the orientation of the cycle with $\overrightarrow{C_{ij}}$ and the set of all arcs $(u,v)$ that are opposite directed with $\overleftarrow{C_{ij}}$.
	\end{itemize}
	See~\Cref{fig:uniqueCycleCij2} for an illustration. 
\end{definition}

\begin{figure}
	
	\centering
	\tikzstyle{vertex}=[circle,fill=black,draw=black,minimum size=8pt,inner sep=0]
	\tikzstyle{edge} = [draw,thick,->]
	\tikzstyle{weight} = []
	\tikzstyle{selected vertex} = [vertex, fill=red!24]
	\tikzstyle{selected edge} = [draw,line width=5pt,-,green!50]
	\usetikzlibrary{arrows,automata}
	\pgfdeclarelayer{background}
	\pgfsetlayers{background,main}
	\begin{tikzpicture}[->,shorten >=1pt,auto,node distance=2.8cm,
	semithick]
	\foreach \pos/\name/\label in {{(-3.5,-1)/r/r},{(0,0)/2/}, {(2,0)/3/i}, {(0,-2)/4/},
		{(-1.5,-1)/1/},{(-4.5,-2)/6/},{(4,-1)/7/}}
	\node[vertex] (\name)[label=$\label$] at \pos {};
	\node[vertex] (5)[label={[yshift=-0.8cm] $j$}] at (2,-2) {}; 
	
	\path (1) edge[<-,thick]   	node[above] {$e_3$} (2)
	(2) edge[->,thick]   	node[above] {$e_4$} (3)
	(3) edge[->,thick,dashed]   	node[right,yshift=0.2cm] { }  node[right]{$e_0$}(5)
	(4) edge[->,thick]   	node[below] {$e_1$} (5)
	(4) edge[->,thick]   	node[below] {$e_2$} (1)
	(r) edge[->,thick]   	node[below] {} (1)
	(6) edge[<-,thick]   	node[below] {} (r)
	(7) edge[->,thick]   	node[below] {} (5);
	
	\begin{pgfonlayer}{background}
	\path  (1) edge[draw,line width=5pt,-,yellow!50] (4)
	(3) edge[draw,line width=5pt,-,yellow!50] (5)
	(2) edge[draw,line width=5pt,-,yellow!50] (3)
	(5) edge[draw,line width=5pt,-,yellow!50] (4)
	(1) edge[draw,line width=5pt,-,yellow!50] (2);
	\end{pgfonlayer}

	\draw[<-,thick] ([shift=(3:1mm)]1.2,-1) arc (0:310:5mm);

	\end{tikzpicture}
	\caption{Example of the unique cycle $C_{e_0}$ induced by $e_0$. Non-tree arcs are drawn as dashed lines. 
		In this example  $e_0\in L$  and therefore 
		$\protect\overrightarrow{C_{e_0}} = \{e_0, e_2, e_4\}$
		and		
		$\protect\overleftarrow{C_{e_0}} =  \{e_1, e_3\}$.
	}
	\label{fig:uniqueCycleCij2}
\end{figure}
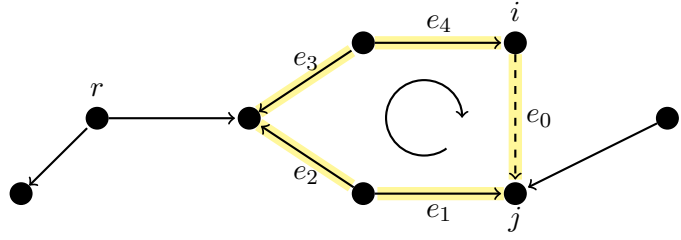

\begin{property}[\cite{Cook1998}]\label{prop:CycleCost=ReducedCost}
	Given a tree structure $(T, L, U)$, the unique cycle $C_{ij}$ induced by an arc $(i,j) \not\in T$ satisfies
	\begin{itemize}
		\item For all arcs $(u,v)\in C_{ij}$ it holds that $(u,v) \in T\cup (i,j)$.
		\item The cost of $C_{ij}$ is given by $c(C_{ij})=\overline{c}_{ij}$ if $(i,j)\in L$ and $c(C_{ij})=-\overline{c}_{ij}$ if $(i,j)\in U$.
	\end{itemize}
\end{property} 

The general idea for the upper and lower bound is based on the fact that we can represent any undirected cycle $C_D$ in $ D $ as a composition  of some unique cycles $ C_{ij} $ induced by  arcs $ (i,j) \in C_D \setminus T $. 
This result is well-known and the base for several algorithms for finding a so-called fundamental set of cycles of a graph, see, e.g.,~\cite{Welch66}.
Here, we define a composition of two sets of arcs as the symmetric difference of the two sets $$ A \circ B = (A \cup B) \backslash (A \cap B). $$

\begin{thm} \label{AnyCycleFrominducedCycles}
	Each undirected cycle $ C_D$ in $D$ can be represented by a composition  of the unique induced cycles $$ C_{ij} \text{ for all } (i,j)\in C_D \setminus T. $$ 
\end{thm}

The proof of~\Cref{AnyCycleFrominducedCycles} can be found in the appendix.

Any cycle $C\in D_f$ yields an undirected cycle $C_D$ which can be represented by a composition of the unique  cycles  $ C_{ij} \text{ for all } {i,j}\in C_D \setminus T $.
We can extend this result by showing that the incidence vector and the cost are closed under this composition. For simplicity, we will denote an arc $(i,j) \notin T$ by $a$ in the following.

\begin{lem}{\label{lem:DirectedComposition}}
	For any cycle $C\in D_f$ let $C_D$ be the undirected equivalent cycle in $D$. Then, it holds that:
	$$ \chi(C) = \sum_{a\in C_D\setminus T} \chi(C_a) \qquad\text{ and } \qquad c(C,f) = \sum_{a\in C_D\setminus T} c(C_a),$$
	where $\chi(C_a)\in \{-1,0,1\}^A$ is defined by 
	\begin{align*}
	\chi_{ij}(C_{a}) := \begin{cases}
	1, &\text{if $(i,j)\in \overrightarrow{C_{a}},$ }\\
	-1, &\text{if $(i,j)\in \overleftarrow{C_{a}},$}\\
	0, &\text{otherwise}\\
	\end{cases} && \text{for all $(i,j) \in A$.}
	\end{align*}
\end{lem}

Again, the proof of~\Cref{lem:DirectedComposition} can be found in the appendix. 
Now we want to show that we can express any optimal flow by an initial optimal tree solution and a conical combination of certain $C_a$.
Let $f$ be an optimal tree solution. \Cref{prop:flowDecomposition} gives us that any feasible integer flow $f'$ can be written as $f' =f + \sum_{i=1}^k \mu_i \chi(C_i)$ for some $C_i\in D_f$ and integers $\mu_i>0$. \Cref{lem:DirectedComposition} yields
$$f' = f + \sum_{i = 1}^k \mu_i \chi(C_i)  = f + \sum_{i = 1}^k \mu_i \sum_{a \in C_{i}\setminus T}\chi(C_a). $$ 
Rearranging the sums and summarizing the cycles and coefficients yields
$$ f' = f + \sum_{a \notin T} \lambda_a \chi(C_a).$$

Here, each cycle $ C_{a} $ is unique and different from the other induced cycles $ C_{a^{\prime}} $ with $a \neq a^{\prime} $. 
Because the cycle $C_a$ only consists of tree arcs and the arc $ a \notin T $ (\Cref{prop:CycleCost=ReducedCost}), $a$ is not included in any of the other induced cycles $ C_{a^{\prime}} $ with $ a \neq a^{\prime} $. 
Hence, since $f'$ is a feasible flow, it must hold that $0\le\lambda_a\le u_a$.
Otherwise, the flow value $f_a$ of the arc $a$ would not satisfy its lower or upper boundary. 
This yields a maximum of $(u_a+1)$ different values for $\lambda_a$. 

Suppose now that $f'$ is an optimal integer solution, then it holds that $c(f')=c(f)$. Again by~\Cref{prop:flowDecomposition} and~\Cref{lem:DirectedComposition} we can show analogously that
$$c(f')=c(f) + \sum_{a \notin T} \lambda_a c(C_a).$$ 
Due to the optimality of $f$ it holds that $c(C_a)\ge 0$ and therefore we have $\lambda_a=0$  for all arcs $a\notin T$ with $c(C_a)>0$. Let the set $S$ consist of all non-tree arcs with zero reduced cost, i.e.,  $$S = \{a\notin T \, | \, \overline{c}_a = 0 \}.$$
Then with~\Cref{prop:CycleCost=ReducedCost}, we can reformulate the above equation for any other optimal integer flow $f'$ and an initial optimal tree solution $f$ as  $$ f' = f + \sum_{a \in S} \lambda_a \chi(C_a).$$

We formalize the above argumentation in~\Cref{theo:optimalFCompositionOfCa} in the appendix.
Moreover, we can show the following:
\begin{lem}\label{lem318}
	Let $f$ be an optimal tree solution, $f^1=f+\sum_{a\in S} \lambda^1_a \chi({C_a})$ and $f^{2}=f+\sum_{a\in S}^k \lambda^{2}_a\chi({C_a})$. If  $\lambda^1\neq \lambda^{2}$ then it holds  $f^1\neq f^{2}$. 
\end{lem}
\begin{proof}
	Let $a' \in S$ be an arc, such that $\lambda^1_{a'}\neq \lambda^{2}_{a'}$. 
	Since $a'$ is contained in the cycle $C_{a^{\prime}}$  and cannot be contained in any other induced cycle, it follows that
	$f^{1}=f+\sum_{a\in S} \lambda^{1}_a\chi({C_a})$
	and $f^{2}=f+\sum_{a\in S} \lambda^{2}_a\chi({C_a})$ have different flow values on the arc $a^{\prime}$. 
	So we have  $f^1_{a^{\prime}}\neq f^{2}_{a^{\prime}}.$ 
\end{proof}

With the above considerations, any  optimal integer solution can be obtained from an initial optimal tree solution and a conical combination of incidence vectors of all zero cost induced cycles with bounded integer coefficients. 
Since any $\lambda_{a}$ for an arc $a\in S$ can take at most $(u_{a}+1)$ different values this gives the following result.

\begin{thm}\label{theo:upperBound}
	Let $F$ be the number of all optimal integer flows in $D$. Then $$
	F\le \max(1, \prod_{a\in S} (u_{a}+1)). $$\qed
\end{thm}

We also can give a lower bound for $F$.  Let $ f $ be an optimal flow. 
Considering one induced cycle $ C_a$ with $a\in S$, we can generate at least as many new optimal flows as the maximum free capacity of the cycle, denoted with $u(C_a)$ allows. Each unit of flow that is added or removed from the cycle yields another flow. 
Considering another induced cycle $C_{a'}$ instead of $C_a$ gives $u(C_{a'})$ many different flows.
Since we can augment over each induced cycle separately, we get the following lower bound for $F$.

\begin{thm}\label{theo:lowerBound}
	Let $F$ be the number of all optimal integer flows in $D$. Then $$ \min (1,\sum_{a \in  S}u(C_a)) \leq  F.$$ 
\end{thm}
With the same argumentation, we can give the following bounds for the total number of feasible integer flows, denoted by \textit{Fea}. Here, we do not have to restrict ourselves to induced cycles with zero reduced cost.
$$ 
\min (1,\sum_{a\notin T} u(C_a)) \le \textit{Fea} \le \max(1,\prod_{a \notin T}(u_{a}+1)).
$$

\begin{figure}
	\centering
	
	\begin{minipage}{0.75\textwidth}
		\subfloat[]{%
			\centering%
			\scalebox{1}{\begin{tikzpicture}[->,shorten >=1pt,auto,node distance=2.8cm,
				semithick]
				\foreach \pos/\name/\dir/\xshift/\yshift in {{(0,0)/a//0cm/0cm}, {(2,0)/b//0cm/0cm}, {(0,-2)/c/left/-0.15cm/-0.1cm},
					{(2,-2)/d/right/0.15cm/-0.1cm}, {(1,-3.5)/e/below/0cm/-0.35cm}}
				\node[vertex,label={[\dir,xshift= \xshift, yshift= \yshift] $\name$}] (\name) at \pos {};
				
				\path (a) edge[->,thick,font=\footnotesize]   	node[] {(0,1,20)} (b)
				(c) edge[<-,thick,font=\footnotesize]   	node[] {(0,1,50)} (a)
				(a) edge[->,thick,font=\footnotesize]   	node[sloped] {(0,4,0)} (d)
				(b) edge[->,thick,font=\footnotesize]   	node[] {(5,5,0)} (d)
				(c) edge[->,thick,font=\footnotesize]   	node[sloped,below] {(0,10,0)} (d)
				(e) edge[<-,thick,font=\footnotesize]   	node[sloped,below] {(12,14,0)} (c)
				(d) edge[->,thick,font=\footnotesize]   	node[sloped,below] {(2,14,0)} (e);

				\begin{pgfonlayer}{background}
				\foreach \source / \dest in {d/a,d/b,c/e,d/e}
				\path[selected edge] (\source.center) -- (\dest.center);
				\end{pgfonlayer}

				\end{tikzpicture}}%
		}
		\hfill
		\subfloat[]{%
			\centering%
			\scalebox{1}{\begin{tikzpicture}[->,shorten >=1pt,auto,node distance=2.8cm,
				semithick]
				\foreach \pos/\name/\dir/\xshift/\yshift in {{(0,0)/a//0cm/0cm}, {(2,0)/b//0cm/0cm}, {(0,-2)/c/left/-0.15cm/-0.1cm},
					{(2,-2)/d/right/0.15cm/-0.1cm}, {(1,-3.5)/e/below/0cm/-0.35cm}}
				\node[vertex,label={[\dir,xshift= \xshift, yshift= \yshift] $\name$}] (\name) at \pos {};
				
				\path (a) edge[->,thick,font=\footnotesize]   	node[] {(0,1,20)} (b)
				(c) edge[<-,thick,font=\footnotesize]   	node[] {(0,1,50)} (a)
				(a) edge[->,thick,font=\footnotesize]   	node[sloped] {(0,4,0)} (d)
				(b) edge[->,thick,font=\footnotesize]   	node[] {(5,5,0)} (d)
				(c) edge[->,thick,font=\footnotesize]   	node[sloped,below] {(0,10,0)} (d)
				(e) edge[<-,thick,font=\footnotesize]   	node[sloped,below] {(2,2,0)} (c)
				(d) edge[->,thick,font=\footnotesize]   	node[sloped,below] {(2,2,0)} (e);

				\begin{pgfonlayer}{background}
				\foreach \source / \dest in {d/a,d/b,c/e,d/e}
				\path[selected edge] (\source.center) -- (\dest.center);
				\end{pgfonlayer}

				\end{tikzpicture}}
		}
	\end{minipage}	
	\caption{Two graphs whose arcs are labeled by $(f_{ij},u_{ij},\overline{c}_{ij})$. Arcs in green represent the tree arcs.}\label{Figrue: ExampleNumberOfFlows}
\end{figure}
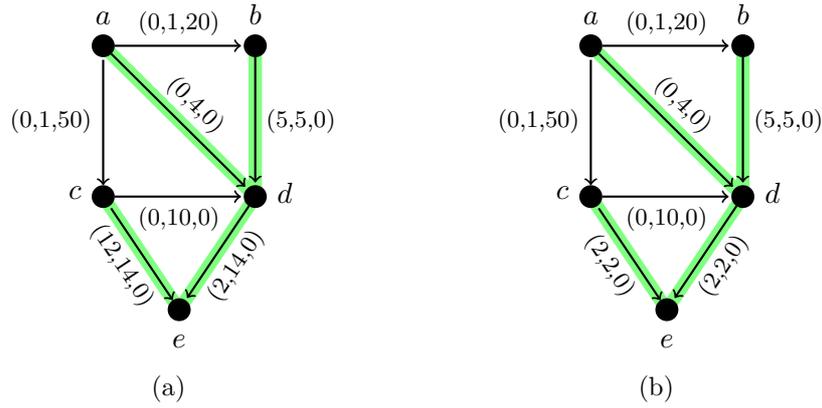
We close this chapter with an example that shows that the bound in~\Cref{theo:upperBound} can be tight, but is not necessarily so.\\
Consider~\Cref{Figrue: ExampleNumberOfFlows}. 
In (a), the arc $(c,d)$ is the only arc not in $T$ with $c_{ij}=0$. We can augment ten times over this arc with one flow unit each time. Since there are no more $ 0 $-cycles, we get $11$ optimal flows, and we have $ | F | = \delta := \max(1, \prod_{a\in S} (u_{a}+1))$.\\
In (b), we have the same $ 0 $-cycle. However, due to the available capacities and flow values on other arcs, we cannot augment over it. 
There is only one optimal flow, but $ \delta = 10 $ and thus $ | F | <\delta $. If we increase the arc $ (c,d) $ capacities  to an arbitrarily large value, we get $ | F | \ll \delta $. \\

\section{Conclusion}
This paper presented the main theoretical bases and implementation of an algorithm to determine all optimal solutions of a linear integer minimum cost flow problem.  
For a given optimal solution, the algorithm efficiently determines proper zero cost cycles by using a depth-first search in a reduced network leading to a new optimal solution. 
The presented algorithm requires $\mathcal{O}(F(m+n)+mn+M)$ time to solve the AOFP.
Using the proposed algorithm, a new method is derived to solve the $K$ best flow problem with an improved running time of $\mathcal{O}(Kn^3+M)$. 
Besides, lower and upper bounds for the number of all optimal integer solutions are shown, which is based on the fact that any optimal integer flow can be obtained from an initial optimal tree solution plus a conical combination of all zero cost induced cycles with bounded coefficients.

Future research could focus on structured combinatorial ways to efficiently explore these conical combinations to derive an improved algorithm for determining alternative solutions in network flow problems. 
An extension of the ideas and techniques presented here to methods for determining supported and non-supported solutions in bi- or multi objective minimum cost flow problems can be another future topic.

%% The Appendices part is started with the command \appendix;
%% appendix sections are then done as normal sections

%% \section{}
%% \label{}

%% If you have bibdatabase file and want bibtex to generate the
%% bibitems, please use
%%
%%  \bibliographystyle{elsarticle-num} 
%%  \bibliography{<your bibdatabase>}

%% else use the following coding to input the bibitems directly in the
%% TeX file.

\appendix
\gdef\thesection{\Alph{section}}
\makeatletter
\renewcommand\@seccntformat[1]{\csname the#1\endcsname.\hspace{0.5em}}
\makeatother
\section{Appendix}

In this section, we want to provide the  missing proofs for some results from~\Cref{UpperBound}.
Recall that any simple path $P$ (cycle $C$) in $D_f$ corresponds to an undirected path $P_D$ (cycle $C_D$) in $D$ and we had defined the incidence vector of $P$ as
\begin{align*}
\chi_{ij}(P) := \begin{cases}
1, &\text{if $P$ traverses $(i,j)$ in $D_f$,}\\
-1, &\text{if $P$ traverses $(j,i)$ in $D_f$,}\\
0, &\text{otherwise}\\
\end{cases} && \text{for all $(i,j) \in A$.}
\end{align*}

We also want to define such a (directed) incidence vector in the same manner for paths that does not necessarily exist in $D_f$.

\begin{definition}
	For an undirected path $P_{u,v}$ from node $u$ to $v$ in $D$, we define $\chi(P_{u,v})\in \{-1,0,1\}^A$  as the \emph{(directed) incidence vector of $P_{u,v}$} as
	\begin{align*}
	\chi_{ij}(P_{u,v}) := \begin{cases}
	1, &\text{if $P_{u,v}$ traverses $(i,j)$ in its forward direction,}\\
	-1, &\text{if $P_{u,v}$ traverses $(i,j)$ in its backward direction,}\\
	0, &\text{otherwise}\\
	\end{cases} && \text{for all $(i,j) \in A$.}
	\end{align*}
	We also define the cost of the path as $c(P_{u,v})= \sum_{(i,j)\in A} \chi_{ij}(P_{u,v}) c_{ij}$, i.e., the cost of an arc traversed in forward direction contributes positively to the total cost, whereas the cost of a backwards-traversed arc contributes negatively.  
	The incidence vector and the cost of a cycle $C$ in $D$ (with a given orientation) are defined analogously.
\end{definition}

These definitions can be seen as also allowing forward and backward arcs in $D_f$ even if the flow value equals the lower or upper capacity, i.e., $A^+:=\{(i,j) \mid (i,j)\in A,\, f_{ij} \leq u_{ij}\}\ \text{ and }
A^- :=\{ (j,i) \mid (i,j) \in A,\, f_{ij}\geq l_{ij}\}$. 
Then, any undirected path $P_D$ corresponds to a directed path $P$ (with not necessarily free capacity) in $D_f$ and  $\chi(P) = \chi(P_D)$.
This leads to the following property.

\begin{property}\label{prop:c(f,P) = c(P_D)}
	Let $P$ ($C$) be a directed path (cycle) in $D_f$ and $P_D$ ($C_D$) its equivalent in $D$. Then $\chi(P) = \chi(P_D)$ ($\chi(C) = \chi(C_D)$) and $c(f,P) = c(P_D)$ ($c(f,C) =c(C_D)$).
\end{property}

With these definitions we can prove the following results:

\begin{thm} [\Cref{AnyCycleFrominducedCycles}]
	Each undirected cycle $ C_D$ in $D$ can be represented by a composition  of the unique induced cycles $$ C_a \text{ for all } a\in C_D \backslash T. $$

\end{thm}
\begin{proof}

	Let $ C_D $ be an undirected cycle in $ D $.
	First, we decompose the cycle $C_D$ into non-tree arcs $a_k=(i_k,j_k)$ and possibly empty tree paths $P^T_{j_l,i_{l+1}}$:
	$$ C_D = \{a_1, P^T_{j_1,i_2}, a_2 , P^T_{j_2,i_3}, \dots, a_t, P^T_{j_t,i_1}   \}$$
	Now we consider the composition of the induced cycles $C_{a_k}$
	\begin{align*}
	&C_{a_1}\circ C_{a_2} \circ \dots \circ C_{a_t} & \\
	=& a_1 \circ P^T_{j_1,i_1} \circ a_2 \circ P^T_{j_2,i_2} \circ \dots \circ a_t \circ P^T_{j_t,i_t} & (C_k = a_k \circ P^T_{j_k,i_k})\\
	=& a_1 \circ P^T_{j_1,r}\circ P^T_{r,i_1} \circ a_2 \circ P^T_{j_2,r} \circ P^T_{r,i_2}\circ \dots \circ a_t \circ P^T_{j_t,r} \circ P^T_{r,i_t} & ( P^T_{j_k,i_k} = P^T_{j_k,r} \circ P^T_{r,i_k})\\
	=&  a_1 \circ P^T_{j_1,r}\circ P^T_{r,i_2}\circ  a_2 \circ P^T_{j_2,r} \circ \dots \circ P^T_{r,i_t} \circ a_t \circ P^T_{j_t,r} \circ P^T_{r,i_1}  & ( \circ \text{ is commutative})\\
	=& a_1 \circ P^T_{j_1,i_2}\circ  a_2 \circ P^T_{j_2,i_3} \circ \dots \circ P^T_{j_{t-1},i_t} \circ a_t \circ P^T_{j_t,i_1}   & (  P^T_{j_k,r}\circ P^T_{r,i_{k+1}} = P^T_{j_k,i_{k+1}})\\
	=& \{a_1, P^T_{j_1,i_2}, a_2 , P^T_{j_2,i_3}, \dots, a_t, P^T_{j_t,i_1}   \}& (\text{all components are disjoint})\\
	=&C_D&
	\end{align*}
	For illustration see~\Cref{Fig: CycleAndCyclceComposition}.

	\qedhere
\end{proof}

\begin{figure}
	\centering
	\subfloat[]{
		\centering
		\scalebox{1}{\begin{tikzpicture}[->,shorten >=1pt,auto,node distance=2.8cm,
			semithick]%
			\foreach \pos/\name/\label\dir in {{(0,0)/2//above}, {(2,0)/3/i_1=j_2/above}, {(0,-2)/4/j_1/below},
				{(-1.5,-1)/1/r/above},{(4,-1)/7/i_2/right/above}}
			\node[vertex] (\name)[label=\dir:$\label$] at \pos {};
			\node[vertex] (5) [label={[yshift=-0.8cm] }] at (2,-2) {};

			\path (1) edge[<-,thick]   	node[above] {} (2)
			(2) edge[->,thick]   	node[above] {} (3)
			(3) edge[->,thick,dashed]   	node[right,yshift=0.2cm] { }  node[right]{}(4)
			(7) edge[->,thick,dashed]   	node[right,yshift=0.2cm] { }  node[right]{}(3)
			(4) edge[->,thick]   	node[below] {} (5)
			(4) edge[->,thick]   	node[below] {} (1)
			(7) edge[->,thick]   	node[below] {} (5);
			
			\begin{pgfonlayer}{background}
			\path  (3) edge[draw,line width=5pt,-,yellow!50] (4)
			(3) edge[draw,line width=5pt,-,yellow!50] (7)
			(5) edge[draw,line width=5pt,-,yellow!50] (7)
			(4) edge[draw,line width=5pt,-,yellow!50] (5);
			\end{pgfonlayer}

			%\node at (0,-2.5) {$j_1$};
			%\node at (2,0.5) {$i_1=j_2$};
			%\node at (4.5,-1) {$i_2$};
			\node at (2,-1) {$C_D$};
			\end{tikzpicture}}%
	}
	
	\begin{minipage}{1\textwidth}
		\subfloat[]{
			\scalebox{1}{\begin{tikzpicture}[->,shorten >=1pt,auto,node distance=2.8cm,
				semithick]%
				\foreach \pos/\name/\label\dir in {{(0,0)/2//above}, {(2,0)/3/i_1=j_2/above}, {(0,-2)/4/j_1/below},
					{(-1.5,-1)/1/r/above},{(4,-1)/7/i_2/right/above}}
				\node[vertex] (\name)[label=\dir:$\label$] at \pos {};
				\node[vertex] (5) [label={[yshift=-0.8cm] }] at (2,-2) {};

				\path (1) edge[<-,thick]   	node[above] {} (2)
				(2) edge[->,thick]   	node[above] {} (3)
				(3) edge[->,thick,dashed]   	node[right,yshift=0.2cm] { }  node[right]{}(4)
				(7) edge[->,thick,dashed]   	node[right,yshift=0.2cm] { }  node[right]{}(3)
				(4) edge[->,thick]   	node[below] {} (5)
				(4) edge[->,thick]   	node[below] {} (1)
				(7) edge[->,thick]   	node[below] {} (5);
				
				\begin{pgfonlayer}{background}
				\path  (3) edge[draw,line width=5pt,-,red!50] (4)
				(3) edge[draw,line width=5pt,-,red!50] (2)
				(1) edge[draw,line width=5pt,-,red!50] (2)
				(4) edge[draw,line width=5pt,-,red!50] (1);
				\end{pgfonlayer}

				%\node at (0,-2.5) {$j_1$};
				%\node at (2,0.5) {$i_1=j_2$};
				%\node at (4.5,-1) {$i_2$};
				\node at (0,-1) {$C_{i_1j_1}$};
				\end{tikzpicture}}%
		}
		\hfill
		\subfloat[]{
			\centering
			\scalebox{1}{\begin{tikzpicture}[->,shorten >=1pt,auto,node distance=2.8cm,
				semithick]
				\foreach \pos/\name/\label\dir in {{(0,0)/2//above}, {(2,0)/3/i_1=j_2/above}, {(0,-2)/4/j_1/below},
					{(-1.5,-1)/1/r/above},{(4,-1)/7/i_2/right/above}}
				\node[vertex] (\name)[label=\dir:$\label$] at \pos {};
				\node[vertex] (5) [label={[yshift=-0.8cm] }] at (2,-2) {};

				\path (1) edge[<-,thick]   	node[above] {} (2)
				(2) edge[->,thick]   	node[above] {} (3)
				(3) edge[->,thick,dashed]   	node[right,yshift=0.2cm] { }  node[right]{}(4)
				(7) edge[->,thick,dashed]   	node[right,yshift=0.2cm] { }  node[right]{}(3)
				(4) edge[->,thick]   	node[below] {} (5)
				(4) edge[->,thick]   	node[below] {} (1)
				(7) edge[->,thick]   	node[below] {} (5);
				
				\begin{pgfonlayer}{background}
				\path  (3) edge[draw,line width=5pt,-,orange!50] (7)
				(3) edge[draw,line width=5pt,-,orange!50] (2)
				(1) edge[draw,line width=5pt,-,orange!50] (2)
				(4) edge[draw,line width=5pt,-,orange!50] (1)
				(5) edge[draw,line width=5pt,-,orange!50] (7)
				(4) edge[draw,line width=5pt,-,orange!50] (5);
				\end{pgfonlayer}

				%\node at (0,-2.5) {$j_1$};
				%\node at (2,0.5) {$i_1=j_2$};
				%\node at (4.5,-1) {$i_2$};
				\node at (2,-1) {$C_{i_2j_2}$};
				\end{tikzpicture}}%
		}
	\end{minipage}	
	\caption{(a) A graph with an (undirected) cycle $C_D$ and (b) (c) the two induced cycles whose composition equals $C_D$. Non-tree arcs are marked as dashed line. It holds $C_D= (i_1,j_1) \circ P^T_{j_1,i_1} \circ (i_2,j_2) \circ P^T_{j_2,i_1}=
		C_{i_1j_1} \circ C_{i_2j_2}$. Here $P^T_{j_2i_1}=\emptyset$.} \label{Fig: CycleAndCyclceComposition}
\end{figure}
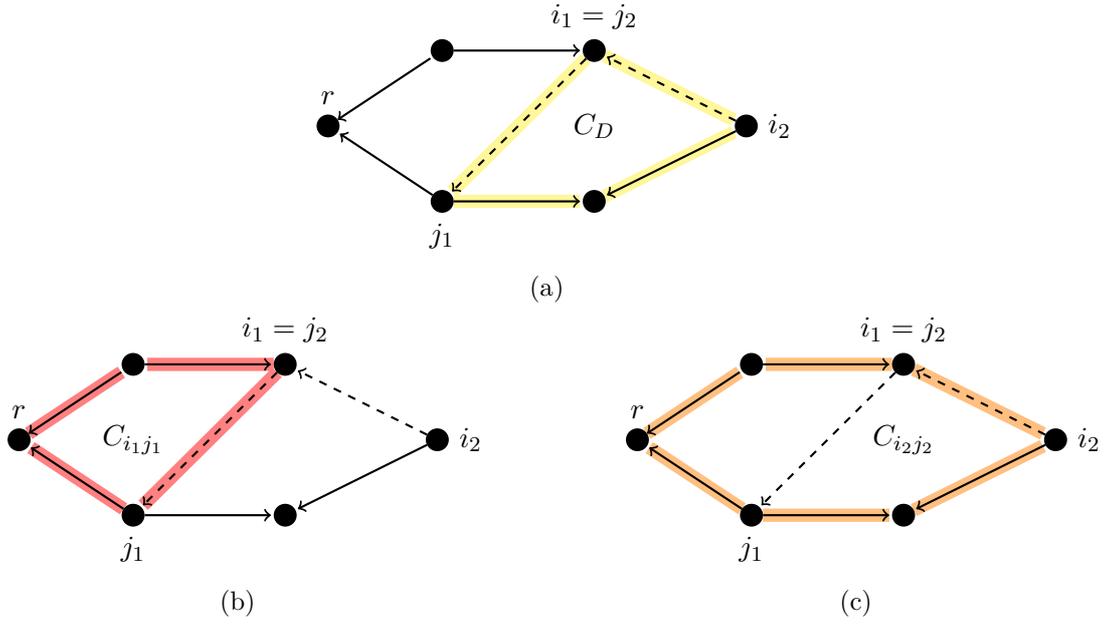

\begin{lem}[\Cref{lem:DirectedComposition}]{\label{lem:DirectedCompositionAppendix}}
	For any cycle $C\in D_f$ let $C_D$ be the corresponding undirected cycle in $D$. Then it holds that:
	$$ \chi(C_D) = \sum_{a\in C_D\backslash T} \chi(C_a) \qquad\text{ and } \qquad c(C_D) = \sum_{a\in C_D\backslash T} c(C_a)$$.
\end{lem} 
\begin{proof}
	Let $C$ be a cycle in $ D_f$, $C_D$ the corresponding cycle in $D$ and let $a$ be a non-tree arc in $\ C_D\setminus T$.
	If $a=(i,j)\in L$, the cycle $C_D$ traverses it in the forward direction, since only $(i,j)$ is present in $D_f$ but $(j,i)$ is not. Similarly, if $a \in U$ then $C_D$ traverses it in its backward direction. So by~\Cref{def:Ca} we have $\chi_{a}(C_D) = \chi_a(C_a) = \sum_{a\in C_D\setminus T} \chi_a(C_a)$.\\
	Now we can follow the previous proof.
	Note that each time the sub-paths canceled each other out by the composition, they could be seen as the same sub-path but traversed in opposite directions. This way, the incidence vectors would sum up to zero for these sub-paths.
	The tree-paths $C_D\cap T$ are traversed in the correct direction by the composition of cycles $C_{a_k}$ due to their correct orientation.\\
	Formally, let $\chi^a(C_a)$ be the vector having only $\chi_a(C_a)$ in the $a$-component, i.e., $\chi^a_a(C_a) = \chi_a(C_a)$ and 0 otherwise.
	So we have $\chi(C_a) = \chi^a(C_a) + \chi(P^T_{j,i})$.
	Again, suppose that \\ $ C_D = \{a_1, P^T_{j_1,i_2}, a_2 , P^T_{j_2,i_3}, \dots, a_t, P^T_{j_t,i_1}   \}$. Then
	\begin{align*}
	&\chi (C_{a_1}) + \chi(C_{a_2}) + \dots + \chi(C_{a_t}) & \\
	=& \chi^{a_1}(C_{a_1}) + \chi( P^T_{j_1,i_1}) + \chi^{ a_2}(C_{a_2}) + \chi(P^T_{j_2,i_2}) + \dots + \chi^{a_t}(C_{a_t}) + \chi(P^T_{j_t,i_t}) & \\
	=& \chi^{a_1}(C_{a_1}) + \chi(P^T_{j_1,r}) +\chi( P^T_{r,i_1}) + \chi^{a_2}(C_{a_2}) + \chi(P^T_{j_2,r}) + \chi(P^T_{r,i_2})&\\
	&+ \dots + \chi^{a_t}(C_{a_t}) + \chi(P^T_{j_t,r}) + \chi(P^T_{r,i_t}) & \\
	=&  \chi^{a_1}(C_{a_1}) + \chi(P^T_{j_1,r}) + \chi(P^T_{r,i_2})+  \chi^{a_2}(C_{a_2}) + \chi(P^T_{j_2,r}) &\\
	&+ \dots + \chi(P^T_{r,i_t}) + \chi^{a_t}(C_{a_t}) + \chi(P^T_{j_t,r}) + \chi(P^T_{r,i_1})  & \\
	=& \chi^{a_1}(C_{a_1}) + \chi(P^T_{j_1,i_2}) +  \chi^{a_2}(C_{a_2}) + \chi(P^T_{j_2,i_3}) + \dots + \chi(P^T_{j_{t-1},i_t}) + \chi^{a_t}(C_{a_t}) + \chi(P^T_{j_t,i_1})   & \\
	=& \chi^{a_1}(C_D) + \chi(P^T_{j_1,i_2}) +  \chi^{a_2}(C_D) + \chi(P^T_{j_2,i_3}) + \dots + \chi(P^T_{j_{t-1},i_t}) + \chi^{a_t}(C_D) + \chi(P^T_{j_t,i_1})& \\
	=&\chi(C_D)&
	\end{align*}
	So we have $ \chi(C_D) = \sum_{a\in C_D\backslash T} \chi(C_a)$.
	To show that the cost is equivalent is now easy:
	\begin{align*}
	c(C_D) =& \sum_{(i,j) \in A} \chi_{ij}(C_D)c_{ij} = \sum_{(i,j) \in A} \sum_{a\in C_D\backslash T}\chi_{ij}(C_a)c_{ij}\\
	=& \sum_{a\in C_D\backslash T} \sum_{(i,j) \in A}  \chi_{ij}(C_a)c_{ij} = \sum_{a\in C_D\backslash T} c(C_a)
	\end{align*}
	
\end{proof}

\begin{thm}\label{theo:optimalFCompositionOfCa}
	Let $f$ be an optimal tree solution with an associated tree structure $(T,L,U)$ and let $f^*$ be another integer flow. Then $f^*$ is optimal if and only if
	$$ f^* = f + \sum_{a \in S} \lambda_a \chi(C_a)$$
	with
	\begin{itemize}
		\item[(i)] $\lambda_a \in \mathbb{Z}$ with $ 0 \leq f_e + \sum_{a\in S} \lambda_a \chi_e (C_a) \leq u_e \quad \forall \, e\in A$,
		\item[(ii)] the set $S$ consists of all non-tree arcs with zero reduced cost: $S = \{a\notin T \, | \, \overline{c}_a = 0 \}$.
	\end{itemize}
\end{thm}
\begin{proof}
	``$\Rightarrow$'': \Cref{prop:flowDecomposition} gives us $ f^* = f + \sum_{i = 1}^k \mu_i \chi(C_i)$ for some $C_i$ in $D_f$ and $\mu_i > 0$.
	With  Property \ref{prop:c(f,P) = c(P_D)}, we can express $f^*$ through the corresponding cycles $C_{D,i}$ in $D$:
	\[
	f^* = f + \sum_{i = 1}^k \mu_i \chi(C_{D,i}) 
	\overset{\text{\Cref{lem:DirectedComposition}}}{=} f + \sum_{i = 1}^k \mu_i \sum_{a \in C_{D,i}\setminus T}\chi(C_a) 
	\]
	Rearranging the sums and summarizing the cycles and coefficients yield
	$$ f^* = f + \sum_{a \notin T} \lambda_a \chi(C_a)$$
	Since $f^*$ is optimal and thereby in particularly feasible $f^*$ satisfies the capacity constraint of all arcs, which is expressed by $(i)$.\\
	The cost of $f^*$ equals the cost of $f$ and we can show in a similar way like above that $c(f^*) = c(f) + \sum_{a \notin T} \lambda_a c(C_a)$.
	Therefore the cost of all considered cycles $C_a$ must be zero. Remember that no cycle can have negative cost because of the negative cycle optimality condition.
	\Cref{prop:CycleCost=ReducedCost} gives us that $c(C_a) = \overline{c}_a$ or $-\overline{c}_a$. So only arcs  $a \in S$ according to $(ii)$ can be considered for representing $f^*$.\\
	``$\Leftarrow$'': Now consider $ f^* = f + \sum_{a \in S} \lambda_a \chi(C_a)$ with $(i)$ and $(ii)$ holding.
	Then $f^*$ satisfies the capacity constraints because of $(i)$.
	It also satisfies the flow balance constraints because we augmented the flow on cycles, so the flow balance at every node remains.
	The cost of each cycle that is used for augmentation is zero because of $(ii)$ and~\Cref{prop:CycleCost=ReducedCost}. So the cost of $f^*$ is $c(f) + \sum_{a \in S} \lambda_a c(C_a) = c(f)$ and therefore $f^*$ is optimal. 
\end{proof}

\end{document}